\documentclass[a4paper,11pt,DIV10,final]{scrartcl}
\usepackage[applemac]{inputenc}
\usepackage{amsmath}
\usepackage{amssymb}
\usepackage{amsthm}
\usepackage{mathtools}
\usepackage{xcolor}
\usepackage{enumerate}
\usepackage{enumitem}
\usepackage[UKenglish]{babel}
\usepackage{ifthen}
\usepackage[notref,notcite,color]{showkeys}
\usepackage{booktabs}
\usepackage{multirow}
\usepackage{xspace}
\usepackage{url}
\usepackage{stmaryrd}

\usepackage[hidelinks,final]{hyperref}

\usepackage{tikz}
\usetikzlibrary{arrows,automata,positioning,snakes}

\definecolor{darkgreen}{rgb}{0.0,0.7,0.0}
\newenvironment{tw}{\noindent\color{darkgreen} TW:}{}

\newenvironment{twwww}{\noindent\color{red} TW:}{}

\newenvironment{mk}{\noindent\color{blue} MK:} {}

\newcommand{\refthm}[1]{Theorem~\ref{#1}\xspace}

\newcommand{\reflem}[1]{Lemma~\ref{#1}\xspace}
\newcommand{\refprop}[1]{Proposition~\ref{#1}\xspace}

\newcommand{\set}[2]{\left\{#1\mathrel{\left|\vphantom{#1}\vphantom{#2}\right.}#2\right\}}
\newcommand{\oneset}[1]{\left\{\mathinner{#1}\right\}}
\newcommand{\smallset}[1]{\left\{#1\right\}}

\newcommand{\geneq}[1]{\dbracket{\hspace*{2pt} #1 \hspace*{2pt}}}

\let\iff=\undefined
\let\implies=\undefined

\newcommand{\iff}       {\mathrel{\Leftrightarrow}}
\newcommand{\implies}   {\text{$\;\Rightarrow\;$}}

\newcommand{\coloneq}   {\mathrel{:=}}

\newcommand{\abs}[1]{\left|\mathinner{#1}\right|}

\newcommand{\dbracket}[1]{\left\llbracket\mathinner{#1}\right\rrbracket} 

\newcommand{\B}{\mathbb{B}}

\newcommand{\logicfont}[1]{\mathrm{#1}}

\newcommand{\FO}{\logicfont{FO}}

\newcommand{\MOD}{\mathrm{MOD}}
\newcommand{\LEN}{\mathrm{LEN}}

\newcommand{\ltrue}      {\ensuremath{\mathord{\top}}\xspace}
\newcommand{\lfalse}     {\ensuremath{\mathord{\bot}}\xspace}
\newcommand{\limplies}   {\mathrel{\rightarrow}}

\newcommand{\liff}       {\mathrel{\leftrightarrow}}

\newcommand{\biglor}     {\mathop{\bigvee}}
\newcommand{\bigland}    {\mathop{\bigwedge}}
\renewcommand{\lor}{\mathrel{\vee}}
\renewcommand{\land}{\mathrel{\wedge}}

\newcommand{\Synt}{\mathrm{Synt}}
\newcommand{\stab}{(\hspace*{-0.4pt}s\hspace*{-0.4pt})}
\newcommand{\greenfont}[1] {\ensuremath{\mathcal{#1}}}

\newcommand{\greenR} {\greenfont{R}\xspace}
\newcommand{\Requiv} {\mathrel{\greenR}}
\newcommand{\Req}    {\Requiv}
\newcommand{\Rleq}   {\mathrel{\leq_\greenR}}
\newcommand{\Rl}     {\mathrel{<_\greenR}}

\newcommand{\greenRs} {{\greenfont{R}^{\stab}}\xspace}
\newcommand{\Requivs} {\mathrel{\greenRs}}
\newcommand{\Reqs}    {\Requivs}
\newcommand{\Rleqs}   {\mathrel{\leq_\greenRs}}

\newcommand{\greenL} {\greenfont{L}\xspace}
\newcommand{\Lequiv} {\mathrel{\greenL}}
\newcommand{\Leq}    {\Lequiv}
\newcommand{\Lleq}   {\mathrel{\leq_\greenL}}
\newcommand{\Ll}     {\mathrel{<_\greenL}}

\newcommand{\greenLs} {{\greenfont{L}^{\stab}}\xspace}
\newcommand{\Lequivs} {\mathrel{\greenLs}}
\newcommand{\Leqs}    {\Lequivs}
\newcommand{\Lleqs}   {\mathrel{\leq_\greenLs}}

\newcommand{\greenJ} {\greenfont{J}\xspace}

\newcommand{\greenJs} {{\greenfont{J}^{\stab}}\xspace}
\newcommand{\Jequivs} {\mathrel{\greenJs}}
\newcommand{\Jeqs}    {\Jequivs}
\newcommand{\Jleqs}   {\mathrel{\leq_\greenJs}}

\renewcommand{\phi}{\varphi}

\newcommand{\mediumSize}[1]{\fontsize{9pt}{12pt}\selectfont #1\normalsize}
\newcommand{\mediumFont}[1]{\normalfont\mediumSize{#1}}
\newcommand{\malcev}%
  {\mathop{\text{\normalsize{\raisebox{0.3mm}{\textcircled{\raisebox{0.1mm}{\mediumFont{m}}}}}}}}
\newcommand{\varietyFont}[1]{\mathrm{\mathbf{#1}}}
\newcommand{\Jone}{\varietyFont{J_1}}

\newcommand{\QL}{\varietyFont{V}}
\newcommand{\QR}{\varietyFont{W}}

\newcommand{\DA}{\varietyFont{DA}}
\newcommand{\VMOD}{\varietyFont{M\hspace*{-0.5pt}O\hspace*{-0.5pt}D}}

\newcommand{\ie}{\textit{i.e.}}
\newcommand{\eg}{\textit{e.g.}}
\newcommand{\cf}{\textit{cf.}}

 \newtheorem{theorem}{Theorem}[section]
 
 \newtheorem{proposition}[theorem]{Proposition}
 \newtheorem{lemma}[theorem]{Lemma}
 \newtheorem{corollary}[theorem]{Corollary}

\newcommand{\eex}{\hspace*{\fill}\ensuremath{\Diamond}}

\newtheorem{expl}[theorem]{Example}
\newenvironment{example}[1][]{\ifthenelse{\equal{#1}{}}{\begin{expl}\upshape}{\begin{expl}[#1]\upshape}}{\eex\end{expl}}
\setlist{itemsep=2pt,parsep=0pt,topsep=3pt}

\title{One Quantifier Alternation in First-Order Logic with Modular Predicates}

\author{Manfred Kuf\-leitner and Tobias Walter}

\publishers{University of Stuttgart, FMI, Germany\,\thanks{The first author was
    supported by the German Research Foundation (DFG) under grant
    \mbox{DI 435/5-1}.}\\
  \texttt{\{kufleitner,walter\}@fmi.uni-stuttgart.de}}

\date{}

\begin{document}

\maketitle

\begin{abstract}
\noindent
\textbf{Abstract.}
Adding modular predicates yields a generalization of first-order
logic~$\FO$ over words. The expressive power of $\FO[{<},\MOD]$ with
order comparison~$x<y$ and predicates for $x \equiv i \bmod n$ has
been investigated by Barrington, Compton, Straubing and Th{\'e}rien. The study of
$\FO[{<},\MOD]$-fragments was initiated by Chaubard, Pin and Straubing. 
More recently, Dartois and \mbox{Paperman} showed
that definability in the two-variable fragment $\FO^2[{<},\MOD]$ is
decidable. In this paper we continue this line of work.

We give an effective algebraic characterization of the word languages
in $\Sigma_2[{<},\MOD]$. The fragment $\Sigma_2$ consists of
first-order formulas in prenex normal form with two blocks of
quantifiers starting with an existential block. In addition we show
that $\Delta_2[{<},\MOD]$, the largest subclass
of~$\Sigma_2[{<},\MOD]$ which is closed under negation, has the same
expressive power as two-variable logic $\FO^2[{<},\MOD]$.  This generalizes the result
$\FO^2[{<}] = \Delta_2[{<}]$ of Th{\'e}rien and Wilke to
modular predicates. As a byproduct, we obtain another decidable
characterization of $\FO^2[{<},\MOD]$.
\end{abstract}

\section{Introduction}

A famous result of McNaughton and Papert says that a language $L$ is definable in first-order logic $\FO[{<}]$ if and only if $L$ is star-free~\cite{mp71}. By a theorem of Sch\"utzenberger, $L$ is star-free if and only if its syntactic monoid is aperiodic~\cite{sch65sf:short}. Therefore, since the syntactic monoid is effectively computable and since aperiodicity of finite monoids is decidable, one can verify whether or not a given regular language is definable in $\FO[{<}]$. Not every regular language is definable in first-order logic $\FO[{<}]$. In particular, one cannot express group properties such as the words of even length. Verifying whether the length is even corresponds to counting modulo $2$. One can think of several ways of adding modular counting modalities to first-order logic. The two most common options are \emph{modular quantifiers} and \emph{modular predicates}. Modular quantifiers yield the logic $\FO{+}\MOD[{<}]$, and Straubing, Th\'erien and Thomas have shown that definability in $\FO{+}\MOD[{<}]$ is decidable~\cite{stt95IC:short}, see also~\cite{EsikLarsen03rairo} for a more general setting. The expressive power of first-order logic $\FO[{<},\MOD]$ with modular predicates was investigated by Barrington, Compton, Straubing and Th{\'e}rien~\cite{BarringtonCST92jcss}. They gave an effective characterization of 
the $\FO[{<},\MOD]$-definable languages.

There are several reasons for the study of fragments of first-order logic. With respect to many computational aspects such as the inclusion problem or the satisfiability problem, first-order logic is non-elementary~\cite{Sto74}. On the other hand, for many interesting properties, one does not require the full expressive power of~$\FO[{<}]$.
For example, when considering the two-variable fragment $\FO^2[{<}]$, then satisfiability is in NP~\cite{wei11phd}. From a very general point of view, the study of fragments also helps with the understanding of all regular languages since they often reveal important characteristics of regular languages (which can be present or absent). For example, one such property is the existence of non-trivial groups in the syntactic monoid.
In addition, fragments give rise to a descriptive complexity theory inside the regular languages: The easier the formalism for defining a given language $L$, the easier is $L$. In the investigation of a fragment~$\mathcal{F}$ several questions arise:
\begin{enumerate}
\item How can one decide whether a given regular language is definable in $\mathcal{F}$? For example, $L$ is definable in $\FO[{<}]$ if and only if its syntactic monoid is aperiodic.
\item Which languages are definable in $\mathcal{F}$? For example, $\FO[{<}]$ defines precisely the star-free languages.
\item Which other fragment defines the same languages as $\mathcal{F}$? For example, three variables are sufficient for defining any $\FO[{<}]$-language~\cite{kam68:short}, \ie, $\FO[{<}]$ and $\FO^3[{<}]$ have the same expressive power.
\item Which closure properties do the $\mathcal{F}$-definable languages have? For example, the $\FO[{<}]$-definable languages are closed under inverse homomorphisms.
\item What is the complexity of the decision and computation problems for $\mathcal{F}$?
\end{enumerate}
In this paper, we are mainly interested in the first three questions. The fourth question can frequently be answered by a result of Lauser and the first author~\cite{KufleitnerL12icalp:short}. Usually logical fragments are defined by restricting some resources in a formula. Typical resources are the number of variables, the quantifier depth, the alternation depth, or the possible atomic predicates. Inside $\FO[{<}]$, for every fixed quantifier depth and every fixed alphabet one can only define a finite number of languages. Therefore, all of the above questions become trivial in this case.
Let~$\Sigma_n$ be the set of all first-order formulas in prenex normal form with at most~$n$ blocks of quantifiers such that the first block is existential, let~$\Pi_n$ be the negations of $\Sigma_n$-formulas, and let $\B\Sigma_n$ be the Boolean closure of $\Sigma_n$.
The fragments~$\Sigma_n$ and $\B\Sigma_n$ define the (quantifier) alternation hierarchy. Over the signature~$[{<}]$, the answer to the second question in case of the alternation hierarchy reveals a surprising connection: A language is definable in $\B\Sigma_n[{<}]$ if and only if it is on the $n$\textsuperscript{th} level of the Straubing-Th\'{e}rien hierarchy~\cite{tho82:short}. The Straubing-Th\'{e}rien hierarchy is an infinite hierarchy exhausting the star-free languages~\cite{str81tcs:short,the81tcs:short}, and it is tightly connected to the dot-depth hierarchy~\cite{str85jpaa}. The fragments $\Sigma_n[{<}]$ correspond to the so-called half levels of the Straubing-Th\'{e}rien hierarchy~\cite{pw97:short}. Decidability criteria are known only for the very first levels of the alternation hierarchy, \ie, for $\Sigma_1[{<}]$, for $\B\Sigma_1[{<}]$ and for $\Sigma_2[{<}]$,~\cite{pin95:short,sim75:short,pw97:short}. Decidability of $\B\Sigma_2[{<}]$ is one of the major open problems in algebraic automata theory.

When restricting the number of variables, then, by Kamp's Theorem~\cite{kam68:short}, using (and reusing) only three variables has the same expressive power as full first-order logic; this fact is often written as $\FO^3[{<}] = \FO[{<}]$. On the other hand, two variables are strictly less powerful. For instance, $(ab)^*$ is definable using three variables but it is not definable in $\FO^2[{<}]$, the two-variable fragment of $\FO[{<}]$. Th\'{e}rien and Wilke have shown that definability in $\FO^2[{<}]$ is decidable and that $\FO^2[{<}]$ and $\Delta_2[{<}]$ have the same expressive power~\cite{tw98stoc:short}. As usual, a language $L \subseteq A^*$ is definable in $\Delta_2[{<}]$ if both $L$ and $A^* \setminus L$ are definable in~$\Sigma_2[{<}]$. This is sometimes written as $\Delta_2[{<}] = \Sigma_2[{<}] \cap \Pi_2[{<}]$.
In particular, $\Delta_2[{<}]$ is the largest subclass of $\Sigma_2[{<}]$ which is closed under complement.

\begin{table}\centering\small
  \begin{tabular}{c@{\hspace*{5mm}}c@{\hspace*{2.5mm}}c@{\hspace*{2.5mm}}c@{\hspace*{2.5mm}}c@{\hspace*{2.5mm}}c@{\hspace*{2.5mm}}c}
    \toprule
    Signature & $\Sigma_1$ & $\B\Sigma_1$ & $\Sigma_2$ & $\FO^2$ & $\FO^2$~\textit{vs.}~$\Delta_2$ & $\FO$
    \\ \midrule
    $[{<}]$ 
    & \parbox{16mm}{\centering decidable \\ \cite{pin95:short}} 
    & \parbox{16mm}{\centering decidable \\ \cite{sim75:short}}
    & \parbox{16mm}{\centering decidable \\ \cite{arf91tcs,pw97:short}} 
    & \parbox{16mm}{\centering decidable \\ \cite{tw98stoc:short}}
    & \parbox{17mm}{\centering equivalent \\ \cite{tw98stoc:short}}
    & \parbox{16mm}{\centering decidable \\ \cite{mp71,sch65sf:short}}
    \\[3.5mm]
    $[{<},\MOD]$ 
    & \parbox{16mm}{\centering decidable \\ \cite{cps06lics:short}}
    & \parbox{16mm}{\centering decidable \\ \cite{cps06lics:short}}
    & \parbox{20mm}{\centering decidable \\ \textbf{new \vphantom{[}result}}
    & \parbox{16mm}{\centering decidable \\ \cite{dartoispaperman13:short}}
    & \parbox{20mm}{\centering equivalent \\ \textbf{new \vphantom{[}result}} 
    & \parbox{16mm}{\centering decidable \\ \cite{BarringtonCST92jcss}}
    \\ \bottomrule
  \end{tabular}
  \smallskip
  \caption{Definability in logical fragments.}
  \label{tab:fragments}
\end{table}

The investigation of fragments over the signature $[{<},\MOD]$ with modular predicates was initiated by Chaubard, Pin and Straubing~\cite{cps06lics:short}. They gave effective algebraic characterizations of $\Sigma_1[{<},\MOD]$- and $\B\Sigma_1[{<},\MOD]$-definability.  Dartois and Paperman~\cite{dartoispaperman13:short} showed that it is decidable whether or not a given regular language is definable in $\FO^2[{<},\MOD]$. In addition, Dartois and Paperman described the languages definable in $\FO^2[{<},\MOD]$. In this paper, we consider the fragment $\Sigma_2[{<},\MOD]$. Our first main result is a decidable algebraic characterization of $\Sigma_2[{<},\MOD]$. As a second result, we show that $\FO^2[{<},\MOD]$ and $\Delta_2[{<},\MOD]$ have the same expressive power. This leads to another decidable characterization of $\FO^2[{<},\MOD]$. As a byproduct, we give a refinement of Dartois and Paperman's  language characterization of $\FO^2[{<},\MOD]$. Our proof technique for $\FO^2[{<},\MOD]$ is different from the one by Dartois and Paperman. It relies on Mal'cev products with definite and reverse definite semigroups. One cannot expect to obtain decidability results for fragments with modular predicates if there is no such result without modular predicates. In this sense, our characterizations complete the picture for the major ``small'' fragments in the presence of modular predicates, see Table~\ref{tab:fragments}. 

Na\"ively, one could expect that modular predicates can only help with expressing group properties, but this is not true. The following example shows that modular predicates increase the expressive power also within the star-free languages.

\begin{example}
The following six languages are all star-free (even though the given expressions for $L_1$ and $L_5$ are using non-trivial star-operations):
\begin{align*}
  L_1 &= \big(\oneset{a,b}^2\big)^*\oneset{aa,bb}\oneset{a,b}^*
  & L_4 &= \oneset{a,b}^*\oneset{ab,ba}\oneset{a,b}^*
  \\ L_2 &= \oneset{a,b}^*aa\oneset{a,b}^*
  & L_5 &= (bc)^*
  \\ L_3 &= \oneset{a,b}^*\oneset{aa,bb}\oneset{a,b}^* 
  & L_6 &= L_3 \cup L_5
\end{align*}
The definability of these languages in the fragments $\Sigma_2$, $\Pi_2$ and $\FO^2$ either with or without modular predicates is depicted in the following diagram:
  
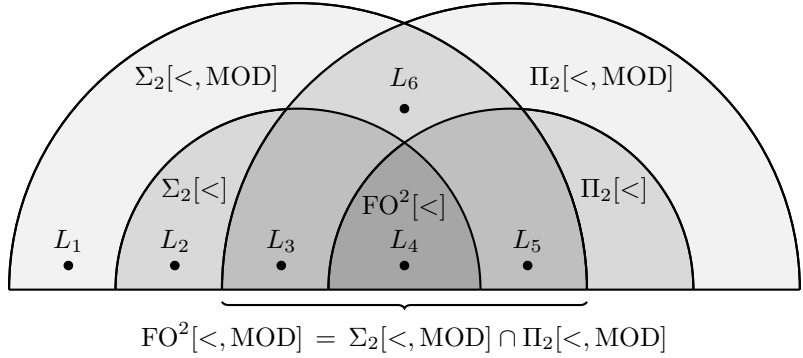
\begin{center}
    \begin{tikzpicture}[scale=0.8]

    \draw[thick,fill=black!5] (1.25,0) arc (180:0:4.75cm);
    
    \draw[thick,fill=black!5] (4.75,0) arc (180:0:4.75cm);
    \draw[thick,fill=black!15] (3,0) arc (180:0:3cm);
    \draw[thick,fill=black!15] (6.5,0) arc (180:0:3cm);

    \begin{scope}
    \clip (4.75,0) arc (180:0:4.75cm);
    \draw[thick,fill=black!15] (1.25,0) arc (180:0:4.75cm);
    \end{scope}
    
    \begin{scope}
    \clip (3,0) arc (180:0:3cm);
    \draw[thick,fill=black!25] (4.75,0) arc (180:0:4.75cm);
    \end{scope}
    
    \begin{scope}
    \clip (6.5,0) arc (180:0:3cm);
    \draw[thick,fill=black!25] (1.25,0) arc (180:0:4.75cm);
    \end{scope}
    
    \begin{scope}
    \clip (3,0) arc (180:0:3cm);
    \draw[thick,fill=black!35] (6.5,0) arc (180:0:3cm);
    \end{scope}

    \draw[thick] (1.25,0) arc (180:0:4.75cm);    
    \draw[thick] (3,0) arc (180:0:3cm);
    \draw[thick] (6.5,0) arc (180:0:3cm);    
    \draw[thick] (4.75,0) arc (180:0:4.75cm);
    
    \draw[thick] (1.25,0) -- (14.25,0);
    \draw (4.3,1.65) node {\small$\Sigma_2[<]$};
    \draw (4.5,3.5) node {\small$\Sigma_2[{<},\MOD]$};
    \draw (11.2,1.65) node {\small$\Pi_2[<]$};
    \draw (11.0,3.5) node {\small$\Pi_2[{<},\MOD]$};
    \draw (7.75,1.4) node {\small$\FO^2[<]$};

   \filldraw[black] (3.975,0.4) circle(2pt);
    \draw (3.975, 0.8) node {\small $L_2$};
    
   \filldraw[black] (2.225,0.4) circle(2pt);
    \draw (2.225, 0.8) node {\small $L_1$};
    
   \filldraw[black] (5.725,0.4) circle(2pt);
    \draw (5.725, 0.8) node {\small $L_3$};
    
   \filldraw[black] (7.75,3) circle(2pt);
    \draw (7.75, 3.5) node {\small $L_6$};
    
   \filldraw[black] (7.75,0.4) circle(2pt);
    \draw (7.75, 0.8) node {\small $L_4$};
    
   \filldraw[black] (9.775,0.4) circle(2pt);
    \draw (9.775, 0.8) node {\small $L_5$};
    
   \draw[snake=brace,mirror snake,thick] (4.75,-0.2) -- node[below,outer sep=4pt] {\small $\FO^2[{<},\MOD]\, = \,\Sigma_2[{<},\MOD] \cap \Pi_2[{<},\MOD]$} (10.75,-0.2);
    
  \end{tikzpicture}
  \end{center}
Examples for the remaining two regions can be obtained by complementation of~$L_1$ and~$L_2$.
Next, we give formulas $\varphi_i$ and $\varphi_i'$ for the languages $L_i$ which justify membership in the respective fragments. We write $\lambda(x)$ for the label of position $x$. For better readability we define the following macros. Let $\mathrm{suc}(x,y) \coloneq x<y \land (\forall z \colon z \leq x \lor y \leq z)$ resemble the successor predicate, the formulas $a(\min) \coloneq \forall z \colon \lambda(z)=a \lor ( \exists x\colon x < z )$ and $a(\max) \coloneq \forall z \colon \lambda(z)=a \lor ( \exists x\colon z < x )$ state that the first (resp.\ last) position in a non-empty word is labeled by $a$, and for letters $a,b$ we set:
\begin{align*}
\varphi_{ab}(x,y) &\coloneq x<y \land \lambda(x) = a \land \lambda(y) = b\\
\psi_{ab}(x,y) &\coloneq \mathrm{suc}(x,y) \land \lambda(x) = a \land \lambda(y) = b
\end{align*}
The formula $\varphi_{ab}(x,y)$ says that $x$ is an $a$-position, $y$ is a $b$-position, and $x$ is smaller than~$y$. The formula $\psi_{ab}(x,y)$ additionally claims that $y = x+1$. We set:
\begin{alignat*}{2}
\varphi_1 &\coloneq \exists x\exists y \colon x \equiv 1 \bmod 2
\land \big(\psi_{aa}(x,y) \lor \psi_{bb}(x,y)\big) &&\in \Sigma_2[{<},\MOD]\\
\varphi_2 &\coloneq \exists x\exists y \colon \psi_{aa}(x,y)&&\in \Sigma_2[<]\\
\varphi_3 &\coloneq \exists x\exists y \colon \psi_{aa}(x,y) \lor \psi_{bb}(x,y) &&\in \Sigma_2[<]\\
\varphi_3' &\coloneq \exists x \exists y \colon 
x \equiv 1 \bmod 2 \land y \equiv 0 \bmod 2
\land \lambda(x) = \lambda(y) &&\in \Pi_2[{<},\MOD] \\
\varphi_4 &\coloneq \exists x \exists y \colon \lambda(x) = a \land \lambda(y) = b&&\in \FO^2[<]\\
\varphi_5 &\coloneq b(\min) \land c(\max) \land
\forall x \forall y\colon \mathrm{suc}(x,y) \limplies \big( \varphi_{bc}(x,y) \lor \varphi_{cb}(x,y) \big) \ &&\in \Pi_2[<]\\
\varphi_5' &\coloneq \LEN^2_0 \land \forall x\colon \lambda(x) \in \smallset{b,c} \land \big(
x \equiv 1 \bmod 2
\,\liff\, \lambda(x) = b\big) \ &&\in \Sigma_2[{<},\MOD]
\end{alignat*}
Note that $\varphi_3', \varphi_5' \in \FO^2[{<},\MOD]$. The formulas for $L_6$ are just the disjunctions of those for $L_3$ and $L_5$. One can show that some language $L_i$ is not definable in some of the above fragments by using the effective algebebraic characterizations of the fragments.
\end{example}

Finally, we remark that the two-variable fragment of first-order logic with modular quantifiers $(\FO{+}\MOD)^2[{<}]$ was characterized by Straubing and Th\'{e}rien~\cite{st03tocs:short}, but there is no immediate connection between the decidability results for the fragments $(\FO{+}\MOD)^2[{<}]$ and $\FO^2[{<},\MOD]$ since $\FO^2[{<},\MOD]$ is not closed under arbitrary inverse homomorphisms. In general, the study of fragments with modular predicates requires so-called $\mathcal{C}$-varieties where $\mathcal{C}$ is the class of length-multiplying homomorphisms, see~\cite{ei03,str02:short,KufleitnerL12icalp:short}.

\section{Preliminaries}
\paragraph*{Words.}
  Let $A$ be a finite alphabet. Elements of $A$ are \emph{letters}. 
We denote by $A^*$ the set of all words over $A$ and by $A^+$ the set of all non-empty 
words over $A$. The empty word is $\varepsilon$. Let $w = w_1w_2w_3$ be a factorization, 
then $w_1$ is a \emph{prefix}, $w_2$ a \emph{factor} and $w_3$ a \emph{suffix} of $w$.
A \emph{language} $L$ is a subset of~$A^*$.
Let $\abs{w}$ denote the \emph{length} of a word $w$ and let $w[i]$ be the letter at position $1\leq i \leq \abs{w}$, 
\ie, we have $w = w[1] \cdot w[2] \cdots w[\abs{w}]$. The \emph{alphabet} of $w$ is the subset $\alpha(w) = \set{w[i]}{1\leq i \leq \abs w}$ of $A$.
Let
\begin{equation*}
 T_n(A) = A\times \oneset{1,\ldots, n}.
\end{equation*}
For $w \in A^*$ we define the word $\tau_{j,n}(w) \in T_n(A)^*$ augmented with some additional information by 
\begin{equation*}
  \tau_{j,n}(w) = (w[1],1+j \bmod n)\cdots(w[\abs{w}], \abs{w}+j \bmod n)
\end{equation*}
and we set $\tau_{n}(w) = \tau_{0,n}(w)$. One can think of $j$ as an offset when counting the positions modulo $n$. 
Words in $T_n(A)^*$ of the form $\tau_{j,n}(w)$ are \emph{well-formed}.
For example,  $\tau_{1,3}(acbabc) = (a,2)(c,3)(b,1)(a,2)(b,3)(c,1)$ is well-formed. Note that by $i \bmod n$ we denote the unique integer $k \in \oneset{1,\ldots,n}$ satisfying $k \equiv i \bmod n$.

\enlargethispage{\baselineskip}

\paragraph*{First-order logic with modular predicates.}

We consider first-order logic $\FO$ interpreted over positions of words. The atomic formulas are 
\begin{equation*}
\ltrue,\,\ \lfalse,\,\ \lambda(x) = a,\,\ x=y,\,\ x < y,\,\ \MOD^n_i(x),\,\ \LEN^n_i\,.
\end{equation*}
The semantics of 
$\ltrue$ is \emph{true}, $\lfalse$ means \emph{false}, $\lambda(x) = a$ states that the position $x$ is labeled by $a$, $x=y$ means $x$ and $y$ are identical, $x < y$ says that the position $x$ is smaller 
than the position $y$, $\MOD^n_i(x)$ holds, if the position $x$ is congruent to~$i$ modulo $n$, and the $0$-ary predicate $\LEN^n_i$ is true 
if the length of the word model is congruent to~$i$ modulo~$n$. Formulas can be composed by the Boolean connectives and 
by existential and universal quantifiers. For better readability, we introduce the following macros:
We write $x \equiv i \bmod n$ for $\MOD^n_i(x)$, we write $x+j \equiv y \bmod n$ for $\bigland_{i=1}^n \big(\MOD^n_i(x) \liff \MOD^n_{i+j}(y)\big)$, and for $B \subseteq A$ we use the shortcut
$\lambda(x) \in B$ for $\biglor_{b\in B} \lambda(x) = b$.
We consider the negation-free $\FO$ fragment $\Sigma_2[{<},\MOD]$ of all formulas without an existential quantifier in the scope of a universal quantifier. Over non-empty words, a formula is in $\Sigma_2[{<},\MOD]$ if there exists an equivalent formula in prenex normal form having two blocks of quantifiers, starting with a block of 
existential quantifiers.  
The fragment $\Pi_2[{<},\MOD]$ contains all negation-free formulas without a universal quantifier in the scope of an existential quantifier. A formula is in $\Pi_2[{<},\MOD]$ if and only if it is equivalent to the negation of a formula in $\Sigma_2[{<},\MOD]$.
By $\FO^2[{<},\MOD]$ we denote the first-order formulas which use only two variables (say $x$ and $y$). We write $\FO^2[{<},\MOD^n]$ for the formulas in~$\FO^2[{<},\MOD]$ 
which use the same modulus $n$ for all modular predicates.
For any class of formulas $\mathcal{F}[{<},\MOD]$ we write $\mathcal{F}[<]$ for the formulas in $\mathcal{F}[{<},\MOD]$ which neither use predicates $\MOD^n_i$ nor $\LEN^n_i$.
A \emph{sentence} is a formula without free variables. For a sentence $\varphi$ we write $u \models \varphi$ if $\varphi$ satisfies $u$. The language defined by a sentence $\varphi$ is $L(\varphi) = \set{u\in A^*}{u\models \varphi}$. Let $\mathcal F$ be a subset of $\FO$. A language $L$ is \emph{definable} in $\mathcal F$ if there 
exists a sentence $\varphi$ in $\mathcal F$ such that $L = L(\varphi)$.
We say that a language is definable in $\Delta_2[{<},\MOD]$ if it is definable in both $\Sigma_2[{<},\MOD]$ and $\Pi_2[{<},\MOD]$. This is often written as $\Delta_2[{<},\MOD] = \Sigma_2[{<},\MOD] \cap \Pi_2[{<},\MOD]$.

\paragraph*{Monoids.}

Let $M$ be a finite monoid. We assume that every finite monoid is equipped with a partial order $\leq$ which is compatible with multiplication, \ie, $x \leq y$ implies $pxq \leq pyq$ for all $p,q \in M$. Note that equality always yields such a partial order. Therefore, ordered monoids generalize the notation of arbitrary monoids.
 An element $e\in M$ is \emph{idempotent} if $e^2 = e$. There exists an integer $\omega \geq 1$ (depending on $M$) such that $x^\omega$ is idempotent for every element $x \in M$.
A \emph{stability index} of a homomorphism $h : A^* \to M$ is a positive integer $s$ such that 
$h(A^s) = h(A^{2s})$. Such numbers exist since $h(A)$ generates an idempotent element in the power monoid 
$\mathcal P(M)$ endowed with the multiplication $XY = \set{xy}{x\in X, y\in Y}$ for $X,Y \subseteq M$; in particular $s = \big(2^{\abs{M}}\big)!$ is a stability index of $h$.
We note that \emph{the} stability index is usually defined as the smallest such number. As all our results hold for every stability index, we refrain from this restriction.
The monoid $S = h((A^s)^*) = h(\smallset{\varepsilon} \cup A^s)$ is called the \emph{stable monoid} of $h$.
Note that $S$ does not dependent on the stability index $s$. For this purpose, let $s'$ be another stability index; 
then $h(A^s) = h(A^{ss'}) = h(A^{s'})$.


\section{Homomorphisms and Recognition}

A homomorphism $h : A^* \to M$ to an ordered monoid $M$ 
\emph{recognizes} a language~$L$ if $L = h^{-1}({\downarrow}h(L))$.
As usual 
${\downarrow}D = \set{x \in M}{\exists y \in D\colon x \leq y}$ for $D \subseteq M$.
 Similarly, we say that $L$ is \emph{recognizable} by a monoid $M$ if there exists a homomorphism $h : A^* \to M$ which recognizes $L$.
A language $L$ is regular if and only if it is recognizable by a finite monoid, see \eg~\cite{pin86}. The \emph{syntactic preorder} $\leq_L$ of a language $L \subseteq A^*$ is defined by $u \leq_L v$ if for all $p,q \in A^*$ the following implication holds:
\begin{equation*}
  pvq \in L \;\implies\; puq \in L\,.
\end{equation*}
We set $u \equiv_L v$ if both $u \leq_L v$ and $v \leq_L u$. The relation $\equiv_L$ is called the \emph{syntactic congruence} of $L$, and the quotient $\Synt(L) = A^* / {\equiv_L}$ is the \emph{syntactic monoid}. The syntactic preorder induces a partial order on $\Synt(L)$ such that the syntactic homomorphism
\begin{align*}
  h_L : A^* &\to \Synt(L) \\
  u &\mapsto \set{v \in A^*}{u\equiv_L v}
\end{align*}
recognizes the language $L$. The syntactic monoid $\Synt(L)$ is the unique minimal monoid which recognizes $L$, see \eg~\cite{pin86}. From any reasonable representation of a regular language $L$ (such as nondeterministic finite automata or sentences in monadic second-order logic) one can effectively compute its syntactic homomorphism $h_L$.

A \emph{positive variety} of finite monoids is a class of finite monoids $\varietyFont{V}$ such that $\varietyFont{V}$ is closed under direct products, submonoids, and monotone homomorphic images. A \emph{full variety} of finite monoids is a positive variety $\varietyFont{V}$ such that $(M,{\leq}) \in \varietyFont{V}$ if and only if $(M,{\geq}) \in \varietyFont{V}$. The order $\geq$ on $M$ is the dual order of $\leq$. Note that $(M,=)$ is a submonoid of the direct product of $(M,\leq)$ and $(M,\geq)$.
A monoid $M$ is \emph{aperiodic} if $x^\omega = x^{\omega+1}$ for all $x \in M$. The class of aperiodic monoids is denoted by $\varietyFont{A}$.
For a monoid $M$ and an idempotent $e \in M$ let $M_e$ be the submonoid of $M$ generated by $\set{a \in M}{e \in MaM}$. A monoid $M$ is in $\DA$ if $eM_e e = e$ for all idempotents $e \in M$, \ie, if $ese = e$ for all $s \in M_e$; see~\cite{DiekertGastinKufleitner08:short,tt02:short} for further characterizations of $\DA$. A monoid $M$ is in $\geneq{x^\omega y x^\omega \leq x^\omega} \malcev \Jone$ if $e M_e e \leq e$ for all idempotents $e \in M$, \ie, if $ese \leq e$ for all $s \in M_e$. 
Usually, one uses relational morphisms for defining Mal'cev products $\varietyFont{W} \malcev \varietyFont{V}$, but in this particular case the current definition is equivalent, \cf~\cite{DiekertGastinKufleitner08:short,pw97:short}.
We have
\begin{equation*} 
  \DA \subseteq \geneq{x^\omega y x^\omega \leq x^\omega} \malcev \Jone \subseteq \varietyFont{A}\,,
\end{equation*}
and membership in each of the classes is decidable.
The classes $\varietyFont{A}$ and $\DA$ form full varieties whereas $\geneq{x^\omega y x^\omega \leq x^\omega} \malcev \Jone$ is a positive variety but not a full variety.


As pointed out by Straubing~\cite{str94}, for $A = \smallset{a,b}$ the syntactic monoids of $L_1 = \set{u \in A^*}{\abs{u} \equiv 0 \bmod 2}$ and $L_2 = \set{u \in A^*}{\abs{u}_a \equiv 0 \bmod 2}$ are both isomorphic to the cyclic group of order $2$. Here, $\abs{u}_a$ denotes the number of occurrences of the letter $a$ in $u$. Since $L_1$ is definable in first-order logic with modular predicates (using the sentence $\LEN^2_0$) whereas $L_2$ is not definable in this logic, the structure of the syntactic monoid cannot be used as a characterization of definability in logical fragments with modular predicates. Instead, we rely on properties of the syntactic homomorphism. Let $\varietyFont{V}$ be a variety of finite monoids. A surjective homomorphism $h : A^* \to M$ is in $\varietyFont{QV}$ if the stable monoid of $h$ is in $\varietyFont{V}$. If membership in $\varietyFont{V}$ is decidable, then, since the stable monoid is effectively computable, membership in $\varietyFont{QV}$ is decidable.
 A language is definable in first-order logic with modular predicates if and only if its syntactic homomorphism is in $\varietyFont{QA}$, \cf~\cite{str94}. Note that in the above example, the stable monoid of $L_1$ is the trivial monoid whereas the stable monoid of $L_2$ is the cyclic group of order two.

Next, we define the class $\varietyFont{V}*\VMOD$. This is usually done in terms of semidirect products of $\varietyFont{V}$ with cyclic groups~\cite{cps06:short}, see also~\cite{cps06lics:short}. In this paper we rely on an equivalent approach using a condition on homomorphisms, see Appendix~\ref{app:semidirect} for a proof of the equivalence. 
The class $\varietyFont{V}*\VMOD$ consists of all surjective homomorphisms $h : A^* \to M$ such that there exists an integer $n>0$ and a homomorphism $g : T_n(A)^* \to N$ with $N \in \varietyFont{V}$ 
satisfying
\begin{equation*}
  g\big(\tau_n(u)\big) \leq g\big(\tau_n(v)\big) 
  \;\Rightarrow\; h(u) \leq h(v) 
\end{equation*}
for all $u,v \in A^*$ with $\abs{u} \equiv \abs{v} \bmod n$.
If $\varietyFont{V}$ is a full variety, then this means that the image $h(u)$ of the word $u \in A^*$ is uniquely determined by the pair $\big(\abs{u} \bmod n, \, g(\tau_n(u))\big)$.
Recall that $T_n(A) = A \times \oneset{1,\ldots, n}$ and $\tau_n(u)$ is the decoration of the word $u$ with positional information modulo $n$. Counting starts at offset $j+1$ when using the notation~$\tau_{j,n}(u)$.

\begin{lemma}\label{lem:strong*MOD}
  Let $\varietyFont{V}$ be a positive variety and let $h : A^* \to M$ be a homomorphism. Suppose there exists an integer $n>0$ and a homomorphism $g : T_n(A)^* \to N$ with $N \in \varietyFont{V}$ such that for all $u,v \in A^*$ with $\abs{u} \equiv \abs{v} \bmod n$ the following implication holds:
\begin{equation*}
  \text{If } g\big(\tau_{j,n}(u)\big) \leq g\big(\tau_{j,n}(v)\big) \text{ for all integers } j \text{, then } h(u) \leq h(v) \text{.}
\end{equation*}
  Then $h$ is in $\varietyFont{V} * \VMOD$.
\end{lemma}

\begin{proof}
  It suffices to consider integers $j \in \oneset{0,\ldots,n-1}$.
  Let $g_j : T_n(A)^* \to N$ be the homomorphism induced by $g_j (a,i) = g(a,i+j \bmod n)$ and let $g' : T_n(A)^* \to \prod_{j=0}^{n-1} N$ be defined by $g'(w) = \big( g_0(w), \ldots, g_{n-1}(w) \big)$. Since we have
  \begin{equation*}
    g'(\tau_n(u)) = \big( g(\tau_{0,n}(u)), \ldots, g(\tau_{n-1,n}(u)) \big)\,,
  \end{equation*}
  this completes the proof.
\end{proof}

A construction which forms the basis of our characterizations of $\Sigma_2[{<},\MOD]$ and $\FO^2[{<},\MOD]$ is the monoid $M_e^{\stab}$. Let $h : A^* \to M$ be a homomorphism with stability index $s$. 
The submonoid $M_e^{\stab}$ of $M$ consists of 
images of words $a_1\cdots a_k$ under $h$ such that $k \equiv 0 \bmod s$ and for every letter $a_i \in A$ there exist words $p_i, q_i$ with $\abs{p_i} \equiv i-1 \bmod s$, $\abs{q_i} \equiv -i \bmod s$ and $h(p_i a_i q_i) = e$. We note that, by definition of the stability index, it suffices to consider words $p_i,q_i$ of length less than $2s$. Therefore, $M_e^{\stab}$ is effectively computable.

\section{The fragment $\mathbf{\Sigma_2}$ with modular predicates}
\label{sec:sigma2mod}

In this section we give an effective algebraic characterization of the first-order fragment $\Sigma_2[{<},\MOD]$ with modular predicates. Without modular predicates, a language $L$ is definable in $\Sigma_2[{<}]$ if and only if its syntactic monoid is in $\geneq{x^\omega y x^\omega \leq x^\omega} \malcev \Jone$, see \eg~\cite{DiekertGastinKufleitner08:short}. We show that a similar result holds, involving submonoids of the form~$\smash{M_e^{\stab}}$ instead of $M_e$. 

\begin{theorem}\label{thm:sigma2mod}
Let $h_L : A^* \to M$ be the syntactic homomorphism of $L \subseteq A^*$ and let $s \geq 1$ satisfy $h_L(A^s) = h_L(A^{2s})$. Then the following conditions are equivalent:
\begin{enumerate}
\item\label{aaa:sigma2mod} $L$ is definable in $\Sigma_2[{<},\MOD]$.
\item\label{bbb:sigma2mod} $L$ is recognized by a homomorphism in $\big(\geneq{x^\omega y x^\omega \leq x^\omega} \malcev \Jone\big) * \VMOD$.
\item\label{ccc:sigma2mod} $h_L$ satisfies $e \smash{M_e^{\stab}} e \leq e$ for all idempotents $e$ in $M$.
\end{enumerate}
\end{theorem}

We give the proof of Theorem~\ref{thm:sigma2mod} in the remainder of this section. We say that a subset~$\mathcal{F}$ of $\FO$ forms a \emph{fragment} if $\mathcal{F}$ is closed under conjunctions and disjunctions and if every atomic formula can be replaced by an arbitrary Boolean combination of atomic formulas. We write $\mathcal{F}[{<},\MOD]$ if arbitrary atomic formulas are allowed whereas $\mathcal{F}[<]$ indicates that only non-modular atomic formulas are considered. In particular, for every fragment $\mathcal{F}[<]$ we write $\mathcal{F}[{<},\MOD]$ for the fragment generated by $\mathcal{F}[<]$ when additionally allowing modular predicates. This notion of fragment is slightly more general than the one introduced in~\cite{KufleitnerL12icalp:short}. A fragment~$\mathcal{F}[<]$ corresponds to a variety $\varietyFont{V}$ if for every language $L$ the following two properties are equivalent: (1) $L$ is definable in $\mathcal{F}$, and (2) its syntactic monoid $\Synt(L)$ is in~$\varietyFont{V}$.

\begin{proposition}\label{prp:fragment2starMOD}
Let $L\subseteq A^*$ and suppose that the fragment $\mathcal{F}[{<}]$ corresponds to the  variety $\varietyFont V$. Then the syntactic homomorphism $h_L$ of $L$ is in $\varietyFont{V}*\VMOD$ if and only if $L$ is a $\mathcal{F}[{<},\MOD]$-definable language.
\end{proposition}

\begin{proof}
Let $\varphi$ be a sentence in $\mathcal F[{<},\MOD]$ which defines $L$.
 We can assume that there is a single integer $n > 0$ such that all modular predicates in $\varphi$ are using the modulus $n$. Let $j \in \smallset{1,\ldots,n}$.
 We replace every occurrence of the atomic predicate $\lambda(x) = a$ in $\varphi$ by $\lambda(x) \in \smallset{a} \times \oneset{1,\ldots,n}$. Similarly, we substitute predicates of the form $x \equiv i \bmod n$ by $\lambda(x) \in A \times \smallset{i}$. Further we replace $\LEN^n_i$ by $\ltrue$ if $i = j$ and by~$\lfalse$ otherwise.
The resulting $\mathcal F[{<}]$-sentence $\varphi'_j$ defines a language $K_j \subseteq T_n(A)^*$. In particular, the syntactic monoid of $K_j$ is in $\varietyFont V$. Let~$g_{j}$ be the syntactic homomorphism of $K_j$ and let $g = \prod_{j=1}^n g_j$ be the homomorphism defined by $g(u) = (g_{1}(u), \ldots, g_{n}(u))$. 
Consider words $u,v \in A^*$ with $\abs{u} \equiv \abs{v} \bmod n$ and  
$g\big(\tau_{i,n}(u)\big) \leq g\big(\tau_{i,n}(v)\big)$ for all integers $i$. 
Suppose $pvq \models \varphi$ and let $j \equiv \abs{pvq} \bmod n$. Then by construction of $\varphi'_j$ we have $\tau_n(pvq) \models \varphi'_j$. Since 
\begin{align*}
g\big(\tau_n(puq)\big) 
&= g\big(\tau_n(p) \tau_{\abs{p},n}(u) \tau_{\abs{pu},n}(q)\big) \\
&\leq g\big(\tau_n(p) \tau_{\abs{p},n}(v) \tau_{\abs{pu},n}(q)\big) 
= g\big( \tau_n(pvq) \big),
\end{align*}
we conclude $\tau_n(puq) \models \varphi'_j$. Again by construction of $\varphi'_j$ we see that $puq \models \varphi$. This shows $h_L(u) \leq h_L(v)$. By Lemma~\ref{lem:strong*MOD} we conclude $h_L \in \varietyFont{V} * \VMOD$.

For the converse let $h_L \in  \varietyFont{V}*\VMOD$. 
Then there exists an integer $n > 0$ and a homomorphism $g : T_n(A)^* \to N$ with $N\in \varietyFont{V}$ such that $g(\tau_n(u)) \leq g(\tau_n(v))$
 implies $h(u) \leq h(v)$
for all $u,v \in A^*$ with $\abs{u} \equiv \abs{v} \bmod n$. 
For every $i \in \oneset{1,\ldots,n}$ we define
\begin{equation*}
  K_i 
  = g^{-1}\Big({\downarrow}g\big(\set{\tau_n(v)}{v\in L, \;
    \abs{v} \equiv i \bmod n}\big)\Big).
\end{equation*}
Since $N \in \varietyFont{V}$ and since $\varietyFont{V}$ corresponds to $\mathcal{F}[{<}]$, there exist formulas 
$\varphi_i' \in \mathcal{F}[{<}]$ with $K_i = L(\varphi_i')$. For every formula $\varphi_i'$ we construct 
$\varphi_i \in \mathcal{F}[{<},\MOD]$ by replacing every atomic proposition $\lambda(x) = (a,j)$ by $\lambda(x) = a \,\land\, x\equiv j \bmod n$. We set $\varphi = \bigvee_i \big(\varphi_i \land \LEN_i^n\big)$. It remains to show $L = L(\varphi)$. 
Consider $u \in A^*$ with $\abs u \equiv i \bmod n$. Then
\begin{align*}
u \in L(\varphi) &\,\iff\, u \in L(\varphi_i) 
\\&\,\iff\, \tau_n(u) \in L(\varphi_i') = K_i
\\&\,\iff\, \exists v\in L\colon g\big(\tau_n(u)\big) \leq g\big(\tau_n(v)\big) \text{ and } \abs v \equiv i \bmod n\\
&\,\iff\, u \in L.
\end{align*}
This shows $L = L(\varphi)$ and thus $L$ is $\mathcal{F}[{<},\MOD]$-definable.
\end{proof}

Proposition~\ref{prp:fragment2starMOD} shows that the first two conditions in Theorem~\ref{thm:sigma2mod} are equivalent.
The characterization in terms of $\big(\geneq{x^\omega y x^\omega \leq x^\omega} \malcev \Jone\big) * \VMOD$ involves some integer $n$ such that positions are counted modulo~$n$. In particular, this characterization of $\Sigma_2[{<},\MOD]$ does not immediately yield decidability.
Roughly speaking, the following lemma implies that counting modulo any stability index is sufficient.

\begin{lemma}\label{lem:starMODtoQS2}
Let $L\subseteq A^*$ be recognizable in $\big(\geneq{x^\omega y x^\omega \leq x^\omega} \malcev \Jone\big) * \VMOD$. Let $h_L : A^* \to M$ be the syntactic homomorphism of $L \subseteq A^*$ and suppose $h_L(A^s) = h_L(A^{2s})$. Then $h_L$ satisfies $e \smash{M_e^{\stab}} e \leq e$ for all idempotents $e$.
\end{lemma}

\begin{proof}
Let $h' : A^* \to M'$ be a homomorphism in $\big(\geneq{x^\omega y x^\omega \leq x^\omega} \malcev \Jone\big) * \VMOD$ which recognizes~$L$. Then there exists an integer $n$ and a homomorphism $g : T_n(A)^* \to N$ with $N \in \geneq{x^\omega y x^\omega \leq x^\omega} \malcev \Jone$ such that for all $u,v \in A^*$ with $\abs{u} \equiv \abs{v} \bmod n$ the following implication holds:
\begin{equation*}
  g\big(\tau_n(u)\big) \leq g\big(\tau_n(v)\big) \; \Rightarrow \; h'(u) \leq h'(v) \text{.}
\end{equation*}
If $n$ is a divisor of $m$, then the homomorphism $\pi : T_m(A)^* \to T_n(A)^*$ induced by the mapping $\pi(a,i \bmod m) = (a, i \bmod n)$ satisfies $\pi(\tau_m(u)) = \tau_n(u)$. 
Therefore, we can assume that $s$ is a divisor of $n$ and that $x^n$ is idempotent for all $x \in N$. 
Let $e \in h_L(A^s)$ be idempotent. 
Consider $a_1 \cdots a_k \in \big(A^s\big)^*$ such that for every letter $a_i \in A$ there exist words $p_i$ and $q_i$ with 
$\abs{p_i} \equiv i-1 \bmod s$, $\abs{q_i} \equiv -i \bmod s$ and $h_L(p_i a_i q_i) = e$. 
Choose $v \in A^s$ such that $h_L(v) = e$. Let $u_i = v^{j_i} p_i a_i q_i v^{j'_i}$ for some integers $j_i,\, j'_i$ 
such that $\abs{v^{j_i} p_i} \equiv i-1 \bmod n$ and $\abs{q_i v^{\smash{j'_i}\vphantom{j}}} \equiv -i \bmod n$. 
We set $u = (u_1 \cdots u_k v^n)^n$. Note that $h_L(u) = e$ and that $g(\tau_n(u)) = f$ is idempotent.
Choose $0 \leq j <n$ with $k + j\abs{v} \equiv 0 \bmod n$. By construction of $u$, we have $\alpha\big(\tau_n(a_1 \cdots a_k v^j)\big) \subseteq \alpha\big(\tau_n(u)\big)$ and thus $g\big( \tau_n(a_1 \cdots a_k v^j) \big) \in N_f$. Since $N \in \geneq{x^\omega y x^\omega \leq x^\omega} \malcev \Jone$, we have 
\begin{equation*}
  g \big( \tau_n(u\, a_1 \cdots a_k v^j u) \big)
  = g \big(\tau_n(u) \tau_n(a_1 \cdots a_k v^j) \tau_n(u) \big) \in f N_f f \leq f = g(\tau_n(u))\,.
\end{equation*}
It follows $h'(u\, a_1 \cdots a_k v^j u) \leq h'(u)$. Suppose $puq \in L$ for some words $p,q \in A^*$. Then $p u a_1 \cdots a_k v^j u q \in L$ since $h'$ recognizes $L$. This shows 
\begin{equation*}
  e \hspace*{1pt} h_L(a_1 \cdots a_k) \hspace*{1pt} e = h_L(u\, a_1 \cdots a_k v^j u) \leq h_L(u) = e\,.
\end{equation*}
Since all elements in $M_e^{\smash{\stab}}$ are of the form $h_L(a_1 \cdots a_k)$ with $a_1 \cdots a_k$ as above, the syntactic homomorphism $h_L$ satisfies $e M_e^{\stab} e \leq e$ for all $e^2 = e$. 
\end{proof}


\begin{lemma}\label{lem:QS2tostarMOD}
Let $h : A^* \to M$ be a homomorphism with stability index~$s$ such that $e M_e^{\stab} e \leq e$ for all $e^2 = e$. Then $h \in \big(\geneq{x^\omega y x^\omega \leq x^\omega} \malcev \Jone\big) * \VMOD$.
\end{lemma}

\begin{proof}
Let $\pi : T_s(A)^* \to A^*$ be the canonical projection. 
We say that a letter $(a,i) \in T_s(A)$ has \emph{offset} $i$.
We define a string rewriting system $\Longrightarrow$ over the alphabet $T_s(A)$ as follows. We set 
$v \Longrightarrow u$ for $u,v \in T_s(A)^+$ if one of the following conditions is satisfied:
\begin{enumerate}
\item $u$ is not well-formed or 
\item both $u$ and $v$ are well-formed, start and end with the same offset and we have $h(\pi(u)) \leq h(\pi(v))$.
\end{enumerate}
Note that $v \Longrightarrow u$ implies $pvq \Longrightarrow puq$. Moreover, $\Longrightarrow$ is reflexive and transitive.
If $v \Longrightarrow u$ with $u$ well-formed, then $v$ is also well-formed.
Let $u \sim v$ if both $u \Longrightarrow v$ and $v \Longrightarrow u$. The relation $\sim$ forms a congruence on $T_s(A)^*$. Every $\sim$-class either contains only well-formed words or it contains only non-well-formed words.
Moreover, there is only one class of non-well-formed words. Every class of nonempty well-formed words is uniquely determined by the offset of the first letter, the offset of the last letter, and the image under $h \circ \pi$. Therefore, the index of $\sim$ is at most $s^2 \abs{M} + 2$; note that the empty word has its own class.
If $u$ is well-formed, then $h(\pi(u)) = h(\pi(v))$ for all words $v$ with $u \sim v$. In particular, the image $h(\pi([u]))$ of a well-formed $\sim$-class $[u]$ is well-defined.
Let $N = {T_s(A)^*}\! /{\sim}$. The relation $\Longrightarrow$ induces a partial order relation $\preceq$ on $N$, \ie, we set $[u] \preceq [v]$ if $v \Longrightarrow u$. By $g : T_s(A)^* \to N$ we denote the natural projection.

Let $f\in N$ be idempotent and let $y \in N_f$. We want to show $fyf \preceq f$.
If $fyf$ is not well-formed, then $fyf \preceq f$ by the first type of rules in the definition of~$\Longrightarrow$. 
Hence we may assume that $fyf$ is a class of well-formed words. 
Since $f^2 = f$, the length of all words in $g^{-1}(f)$ is divisible by $s$. Let $e = h(\pi(f))$ and $x = h(\pi(y))$. The element $e \in M$ is idempotent and we have $x \in M_e^{\stab}$. To see the latter property, suppose $g(b_1 \cdots b_k) = y$ for $b_i \in T_s(A)$. Since all words in $g^{-1}(fyf)$ are well-formed, we have $k \equiv 0 \bmod s$ and $b_1 \cdots b_k$ starts and ends with the same offsets as the words in $g^{-1}(f)$.
Let $p_i, q_i \in T_s(A)^*$ with $g(p_i b_i q_i) = f$. Since $p_i b_i q_i$ is well-formed, we have $\abs{p_i} \equiv i-1 \bmod s$ and $\abs{q_i} \equiv -i \bmod s$. Applying the homomorphism $\pi$ yields $x = h(\pi(b_1 \cdots b_k)) \in M_e^{\stab}$. It follows $exe \leq e$. By definition of $\Longrightarrow$ we conclude $v \Longrightarrow u$ for all $v \in g^{-1}(f)$ and all $u \in g^{-1}(fyf)$. This shows $fyf \preceq f$ as desired. Hence $N \in \geneq{x^\omega y x^\omega \leq x^\omega} \malcev \Jone$.

We finally show that $g(\tau_s(u)) \preceq g(\tau_s(v))$ implies $h(u) \leq h(v)$ for all $u,v \in A^*$ with $\abs u \equiv \abs v \bmod s$.
To this end, we need to prove that 
$\tau_s(v) \Longrightarrow \tau_s(u)$ implies $h(u) \leq h(v)$. However, by definition, this holds since 
$\tau_s(u)$ is well-formed.
\end{proof}

The proof technique used in Lemma~\ref{lem:QS2tostarMOD} is quite general, it works as soon as the (ordered) syntactic monoid is available for recognizing the non-wellformed words.
We can now combine the results in this section to obtain a proof of Theorem~\ref{thm:sigma2mod}.

\begin{proof}[Proof of Theorem~\ref{thm:sigma2mod}]
The first-order fragment $\Sigma_2[{<}]$ corresponds to the positive variety $\geneq{x^\omega y x^\omega \leq x^\omega} \malcev \Jone$, see \eg~\cite{DiekertGastinKufleitner08:short}. Therefore, the equivalence of ``\ref{aaa:sigma2mod}'' and ``\ref{bbb:sigma2mod}'' follows by Proposition~\ref{prp:fragment2starMOD}. The implications from ``\ref{bbb:sigma2mod}'' to ``\ref{ccc:sigma2mod}'' is Lemma~\ref{lem:starMODtoQS2}, and Lemma~\ref{lem:QS2tostarMOD} shows that~``\ref{ccc:sigma2mod}'' implies ``\ref{bbb:sigma2mod}''.
\end{proof}

\enlargethispage{1\baselineskip}

Theorem~\ref{thm:sigma2mod} and its dual version for $\Pi_2[{<},\MOD]$ immediately lead to the following effective characterization of $\Delta_2[{<},\MOD]$-definable languages.

\begin{corollary}\label{cor:delta2mod}
Let $h_L : A^* \to M$ be the syntactic homomorphism of $L \subseteq A^*$ and let $s \geq 1$ satisfy $h_L(A^s) = h_L(A^{2s})$. Then the following conditions are equivalent:
\begin{enumerate}
\item\label{aaa:delta2mod} $L$ is definable in $\Delta_2[{<},\MOD]$.
\item\label{bbb:delta2mod} $h_L$ satisfies $e M_e^{\stab} e = e$ for all idempotents $e$ in $M$.
\end{enumerate}
\end{corollary}

\begin{proof}
The language $L$ is $\Delta_2[{<},\MOD]$-definable if, and only if, it is definable in both $\Sigma_2[{<},\MOD]$ and $\Pi_2[{<},\MOD]$ if, and only if, $e M_e^{\stab} e \leq e$ and $e \leq e M_e^{\stab} e$ for all idempotents $e\in M$ if, and only if, $e M_e^{\stab} e = e$ for all idempotents $e$.
\end{proof}

\section{The fragment \textrm{FO}$\mathbf{^2}$ with modular predicates}
\label{sec:fo2mod}

Dartois and Paperman have shown that a language is definable in two-variable first-order logic $\FO^2[{<},\MOD]$ with modular predicates if and only if its syntactic homomorphism is in $\varietyFont{Q}\DA$~\cite{dartoispaperman13:short}, thereby showing that it is decidable whether or not a given language is definable in $\FO^2[{<},\MOD]$.
The main result of this section establishes a new effective algebraic characterization of $\FO^2[{<},\MOD]$. Since this characterization is the same as the one for $\Delta_2[{<},\MOD]$ in Corollary~\ref{cor:delta2mod}, this immediately implies that $\FO^2[{<},\MOD]$ and $\Delta_2[{<},\MOD]$ have the same expressive power. This extends the result of Th{\'e}rien and Wilke that $\FO^2[<]$ and $\Delta_2[<]$ without modular predicates define the same languages~\cite{tw98stoc:short}.

The equivalence of $\FO^2[{<},\MOD]$ and $\Delta_2[{<},\MOD]$ does not immediately follow from Proposition~\ref{prp:fragment2starMOD} and the Th{\'e}rien-Wilke result for two reasons. First, formally $\Delta_2$ is not a fragment. A typical example which illustrates this problem is the language $L = A^* ab A^*$ defined by the following $\Sigma_2[{<}]$-sentence:
\begin{equation*}
  \varphi \,:=\, \exists x \exists y \forall z\colon x<y \land \lambda(x) = a \land \lambda(y) = b \land (z \leq x \lor y \leq z)
\end{equation*}
If $A = \{a,b\}$, then $L$ is definable in~$\Pi_2[<]$; and if $A = \{a,b,c\}$, then $L$ is not definable in~$\Pi_2[<]$. Therefore, saying whether the sentence $\varphi$ is in $\Delta_2[<]$ is not well-defined.
Second, the operation $\varietyFont{V} \mapsto \varietyFont{V} * \varietyFont{MOD}$ is not compatible with intersection, see Example~\ref{exa:intersectionMOD} below. Therefore, applying Proposition 4.2 to $\Sigma_2$ and $\Pi_2$ separately does immediately yield a characterization of $\Delta_2$.

Dartois and Paperman proved that the languages in $\FO^2[{<},\MOD]$ are exactly the so-called unambiguous modular polynomials. As a byproduct, we refine this result by showing that modular determinism and co-determinism can be used as the sole reason of unambiguity. The proof of our result relies on different techniques than the one by Dartois and Paperman. This new language characterization in terms of modular deterministic and co{-}deterministic products can be seen as an extension of a corresponding result without modular predicates~\cite{kw12}.
Let $L,K \subseteq A^*$ and $a \in A$. The product $LaK$ is \emph{determistic} if every word in~$LaK$ has a unique prefix in $La$. Symmetrically the product $LaK$ is \emph{co-deterministic} if every word in $LaK$ has 
a unique suffix in $aK$. 
We further introduce a special kind of (co-)deterministic products.
The product $LaK$ is \emph{$n$-modularly deterministic} if all words in $L$ have 
the same length $i$ modulo $n$ and $(a,i+1) \not\in \alpha(\tau_n(L))$, \ie, the letter $a$ in the product $LaK$ is the first occurrence at a position congruent $i+1$ modulo $n$. A product is \emph{modularly deterministic} if it is $n$-modularly deterministic for some integer $n \geq 1$.
\emph{Modularly co-deterministic} and \emph{$n$-modularly co-deterministic products} are defined symmetrically.
It is easy to see that modularly (co-)deterministic products are indeed (co-)deterministic.

\enlargethispage{\baselineskip}

\begin{theorem}\label{thm:FO2char}
Let $h_L : A^* \to M$ be the syntactic homomorphism of $L \subseteq A^*$ and let $s \geq 1$ satisfy $h_L(A^s) = h_L(A^{2s})$. Then the following conditions are equivalent:
\begin{enumerate}
\item\label{aaa:FO2char} $L$ is definable in $\FO^2[{<},\MOD]$.
\item\label{bbb:FO2char} $L$ is recognized by a homomorphism in $\varietyFont{DA}*\VMOD$.
\item\label{ccc:FO2char} $h_L$ satisfies $eM^{\stab}_ee = e$ 
for all idempotents $e$ in $M$.
\item\label{ddd:FO2char} $L$ is expressible from languages of the form $(A_1\cdots A_s)^*$ for $A_i \subseteq A$ using disjoint unions and 
$s$-modularly deter\-mi\-nistic and co{-}deter\-mi\-nis\-tic products.
\end{enumerate}
\end{theorem}

For the proof of \refthm{thm:FO2char} we need additional techniques. First, we define Green's relations 
which are a classical tool in semigroup theory.
Let $M$ be a monoid and let $x,y \in M$. We set 
$x \Rleq y$ if $xM \subseteq yM$, and we set $x\Lleq y$ if $Mx \subseteq My$. We define similar notions using the stable monoid. 
Let $h : A^* \to M$ be a homomorphism with stable monoid $S$.
Then we set 
\begin{align*}
x\Jleqs y &\,\,\iff\,\, SxS \subseteq SyS, \\
x\Rleqs y &\,\,\iff\,\, xS \subseteq yS, \\
x\Lleqs y &\,\,\iff\,\, Sx \subseteq Sy.
\end{align*}
Let $\greenfont{G} \in \oneset{\Jeqs,\, \Req,\, \Reqs,\, \Leq,\, \Leqs}$. Then we set $x \mathrel\greenfont{G} y$ if both $x \leq_{\greenfont{G}} y$ and $y \leq_{\greenfont{G}} x$. We write $x <_{\greenfont{G}} y$ if $x \leq_{\greenfont{G}} y$ but not $x \mathrel\greenfont{G} y$. 
A monoid is \emph{$\greenfont{G}$-trivial} if every $\greenfont{G}$-class contains only one element.
It is easy to see that $\leq_{\greenfont G}$ is a preorder and $\greenfont{G}$ is an equivalence relation. 
The relations $\Reqs$ and $\Leqs$ have a similar purpose as the relations $\greenfont{R}_{st}$ and $\greenfont{L}_{st}$ introduced in~\cite{dartoispaperman13:short}, yet they are not the same. If $h$ is the syntactic homomorphism of the language $(A^2)^*$, then $\Reqs$ and $\Leqs$ are the identity relation (since $S = \smallset{1}$) whereas $\greenfont{R}_{st}$ and $\greenfont{L}_{st}$ are universal.

\begin{lemma}\label{lem:RRst}
  Let $h : A^* \to M$ be a surjective homomorphism with stability index $s$. If $x h(u) \Req x$ for $u \in (A^s)^*$, then $x h(u) \Reqs x$. If $h(u) x \Leq x$ for $u \in (A^s)^*$, then $h(u) x \Leqs x$.
\end{lemma}

\begin{proof}
  By left-right symmetry, it suffices to prove the first statement.
  Let $v \in A^*$ such that $x h(uv) = x$. Let $v' = v (uv)^{s-1}$. Then $v' \in (A^s)^*$ and $x h(u v') = x$. This shows $x \Rleqs x h(u)$.
\end{proof}

A typical application of Lemma~\ref{lem:RRst} is in the case of $xe \Req x$ or $ex \Leq x$ for some idempotent $e \in M$ since then we have $e \in S = h\big((A^s)^*\big)$.

\begin{lemma}\label{lem:Jregidem}
Let $h : A^* \to M$ be a homomorphism such that $e M^{\stab}_e e = e$ for all idempotents $e\in M$. 
Then $x \Jeqs e^2 = e$ implies $x^2 = x$.
\end{lemma}

\begin{proof}
Let $S$ be the stable monoid of $h$. 
Consider $u,v \in S$ such that $x = uev$. Since $x \Jeqs e$, there exist $u',v' \in S$ such that $e = u' u e v v'$. This shows $u,v \in M^{\stab}_e$.
It follows $x^2 = u(evue)v = uev = x$.
\end{proof}

\begin{lemma}\label{lem:QDAinduction}
Let $h : A^* \to M$ be a surjective homomorphism with stability index $s$,  let $\pi : M \to N$ be a surjective homomorphism, and let $g = \pi \circ h : A^* \to N$. Then $s$ is also a stability index of $g$; and if $eM^{\stab}_ee=e$ for all idempotents $e\in M$, then we have $f N^{\stab}_f f=f$ for all idempotents $f \in N$.
\end{lemma}

\begin{proof}
We have $g(A^s) = \pi(h(A^s)) = \pi(h(A^{2s})) = g(A^{2s})$. Suppose $eM^{\stab}_ee=e$ for all idempotents $e\in M$. Let $f \in N$ be idempotent and consider a word $a_1 \cdots a_k \in A^*$ with $k \equiv 0 \bmod s$ such that there exist $p_i, q_i \in A^*$ with 
$\abs{p_i} \equiv i-1 \bmod s$, $\abs{q_i} \equiv -i$ and $g(p_ia_iq_i) = f$. With $u_i = p_i a_i q_i$ we define $u = (u_1 \cdots u_k)^n$ for some $n \geq 1$ such that $h(u) = e$ is idempotent. By considering factorizations of $u$ we can choose words $p_i', q_i' \in A^*$ with 
$\abs{p_i'} \equiv i-1 \bmod s$, $\abs{q_i'} \equiv -i$ and $h(p_i'a_iq_i') = h(u) = e$. Therefore $h(u a_1 \cdots a_k u) = e h(a_1 \cdots a_k) e = e = h(u)$. It follows $f g(a_1 \cdots a_k) f = \pi(h(u a_1 \cdots a_k u)) = \pi(h(u)) = f$ which completes the proof.
\end{proof}

The next lemma is an analogue of a basic property of Green's relations.
An $\Reqs$-class is \emph{regular} if it contains an idempotent element.

\begin{lemma}\label{lem:Rregtrivial}
Let $h : A^* \to M$ be a homomorphism 
and let every regular $\Reqs$-class of $M$ be trivial. Then $M$ is $\Reqs$-trivial.
\end{lemma}

\begin{proof}
Let $S$ be the stable monoid of $h$ and let $x \Reqs y$. 
Then there exist $u,v \in S$ such that $xu = y$ and $yv = x$. 
Since $(uv)^\omega \Reqs (uv)^\omega u$ is within a regular $\Reqs$-class, 
we have $(uv)^\omega = (uv)^\omega u$. 
Hence we conclude $x = yv = xuv = x (uv)^\omega =  x (uv)^\omega u =  xu = y$.
\end{proof}

Let $M$ be a monoid. 
For $x,y \in M$ we set $x \sim_K y$ if for every idempotent $e \in M$ we have either $ex = ey$ or $ex, ey \Rl e$. Symmetrically, we 
set $x \sim_D y$ if for every idempotent $e \in M$ we have either $xe = ye$ or $xe, ye \Ll e$. The relations $\sim_K$ 
and $\sim_D$ form congruences~\cite{krt68arbib8:short}.
We define Mal'cev products of the semigroup varieties $\varietyFont{K}$ and~$\varietyFont{D}$ and classes of homomorphisms $\varietyFont{V}$. 
A surjective homomorphism $h : A^* \to M$ onto a finite monoid $M$ is in $\varietyFont{K} \malcev \varietyFont{V}$ if $\pi_K \circ h : A^* \to M/{\sim_K}$ is in~$\varietyFont{V}$. Here, $\pi_K : M \to M/{\sim_K}$ is the natural projection. The definition of $\varietyFont{D}\malcev \varietyFont{V}$ is similar using~$\sim_D$.
The definition of Mal'cev products usually relies on relational morphisms, but the current approach is equivalent~\cite{krt68arbib8:short}.
Let $\QR_2$ be the class of homomorphisms $A^* \to M$ onto $\Reqs$-trivial monoids, let $\QL_2$ be the class of homomorphisms $A^* \to M$ onto $\Leqs$-trivial monoids, and let $\QL_{m+1} = \varietyFont{D}\malcev \QR_{m}$ and $\QR_{m+1} = \varietyFont{K}\malcev \QL_{m}$ for $m\geq 2$.
The definition starts with index $m=2$ in order to match the corresponding levels of the Trotter-Weil hierarchy, \cf~\cite{KufleitnerLauser12mfcs:short}.

The next lemma shows that a homomorphism $h : A^* \to M$ which satisfies $eM_e^{\stab}e = e$ is within this hierarchy.

\begin{lemma}\label{lem:QDA2hierarchy}
Let $h : A^* \to M$ be a homomorphism satisfying $eM^{\stab}_ee = e$ for all idempotents $e$. 
Then $h \in \QR_m \cap \QL_m$ for some $m \geq 2$.
\end{lemma}
\begin{proof}
We can assume that $M$ is either not $\Reqs$-trivial or not $\Leqs$-trivial.
By induction on the number of non-trivial $\Reqs$- and $\Leqs$-classes we show that after finitely many quotients with $\sim_K$ and $\sim_D$ we obtain a homomorphism in $\QR_2 \cap \QL_2$. This induction scheme relies on Lemma~\ref{lem:QDAinduction}.

By left-right symmetry (and using Lemma~\ref{lem:Jregidem}
 and Lemma~\ref{lem:Rregtrivial}) we can assume that there exist two idempotents $f \neq g$ in~$M$ with $f \Reqs g$. Moreover we can choose $f$ and~$g$ such that all regular $\Leqs$-classes which are $<_{\Jeqs}$-below $f$ are trivial (note that this can be achieved either with $f \Reqs g$ and $\Leqs$-classes below $f$ or with $f \Leqs g$ and $\Reqs$-classes below $f$, and the latter situation is left-right symmetric). From $f \Reqs g$ we obtain $fg = g$ and $gf = f$.

We want to show $f \sim_D g$. Consider an idempotent $e \in M$. The proof of $f \sim_D g$ consists of two steps. First, we convince ourselves that $fe \Leq e$ if and only if $ge \Leq e$. Second, we verify that $fe \Leq e \Leq ge$ implies $fe = ge$.

If $fe \Leq e$, then $fe \Leqs e$ by Lemma~\ref{lem:RRst}; therefore we have $f \in M^{\stab}_e$ and $e = efe = egfe$. This shows $g \in M^{\stab}_e$ and $ege = e$, \ie, $ge \Leqs e$. This completes the first step. For the second step, we can assume $fe \Leqs e \Leqs ge$ by  Lemma~\ref{lem:RRst}. 
Since $fe \leq_{\Jeqs} f$, there are two possible cases: Either $fe <_{\Jeqs} f$ or $fe \Jeqs f$. If $fe <_{\Jeqs} f$, then, by the assumption on the $\Leqs$-classes \smash{$<_{\Jeqs}$}-below $f$, we have $fe = e$ and thus $fe = gfe = ge$. If $fe \Jeqs f$, then $e \in M^{\stab}_f$ and $fef = f$. This implies $ge = fge = fefge = fe$. In any case, we have $fe = ge$.
\end{proof}

The following lemma shows that certain information is never destroyed by the congruences $\sim_K$ and $\sim_D$. For example, if $h: A^* \to M$ can distinguish the length of a word modulo $s$, then so does $\pi_K \circ h: A^* \to M/{\sim_K}$. 
This property is used in the induction scheme of \refprop{prp:hierarchy2product}.

\begin{lemma}\label{lem:alphinduction}
Let $h : A^* \to M$ be a homomorphism with stability index $s$ such that 
 $h(u) = h(v)$ implies $\abs{u} \equiv \abs{v} \bmod s$ for all $u,v \in A^*$, and 
$h(u') = h(v')$ implies $\alpha(\tau_s(u')) = \alpha(\tau_s(v'))$ for all $u',v' \in (A^s)^*$.
 Let $\pi : M \to M/{\sim_K}$ be the natural projection and let $g = \pi \circ h : A^* \to M/{\sim_K}$.
Then $g(u) = g(v)$ implies $\abs{u} \equiv \abs{v} \bmod s$ for all $u,v \in A^*$, and 
$g(u') = g(v')$ implies $\alpha(\tau_s(u')) = \alpha(\tau_s(v'))$ for all $u',v' \in (A^s)^*$.
\end{lemma}

\begin{proof}
Let $x^n$ be idempotent for all $x \in M$.
Let $u,v \in A^*$ with $g(u) = g(v)$. This means $h(u) \sim_K h(v)$. We have $h(u^{sn}u) \Req h(u^{sn})$ and hence $h(u^{sn}u) = h(u^{sn}v)$. This implies $\abs{u} \equiv \abs{u^{sn}u} \equiv \abs{u^{sn}v} \equiv \abs{v} \bmod s$.
Similarly, let $u',v' \in (A^s)^*$ with $h(u') \sim_K h(v')$.
Then $h(u'^n u') \Req h(u'^n)$ and thus $h(u'^n u') = h(u'^n v')$. 
This yields $\alpha(\tau_s(v')) \subseteq \alpha(\tau_s(u'^n v')) = \alpha(\tau_s(u'^n u')) = \alpha(\tau_s(u'))$. Symmetrically, we have $\alpha(\tau_s(u')) \subseteq 
\alpha(\tau_s(v'))$.
\end{proof}

\pagebreak[3]

\begin{proposition}\label{prp:hierarchy2product}
Let $h : A^* \to M$ be a surjective homomorphism with stability index $s$ satisfying the following three properties:
\begin{itemize}
\item $eM^{\stab}_e e = e$ for all idempotents $e \in M$,
\item $h(u) = h(v)$ implies $\abs{u} \equiv \abs{v} \bmod s$ for all $u,v \in A^*$, and 
\item $h(u') = h(v')$ implies $\alpha(\tau_s(u')) = \alpha(\tau_s(v'))$ for all $u',v' \in (A^s)^*$.
\end{itemize}
If $L \subseteq A^*$ is recognized by $h$, then $L$ is expressible from the languages $(A_1\cdots A_s)^*$ for $A_i \subseteq A$ using disjoint unions and 
$s$-modularly deter\-mi\-nistic and co{-}deter\-mi\-nis\-tic products.
\end{proposition}

\begin{proof}
Let $\pi : M \to N = M/{\sim_K}$ be the natural projection and let $g = \pi \circ h : A^* \to N$. 
By Lemma~\ref{lem:QDA2hierarchy}, the homomorphism $h$ is in $\QR_m$ for some $m$. If $m > 2$ then we can assume $h \not \in \QL_{m-1}$, since otherwise we proceed with a symmetric construction using $s$-modularly co{-}deterministic products. Therefore, in the case of $m > 2$, we can assume that all $g$-recognizable languages have the desired property.
The homomorphism $g$ satisfies the desired presumptions by \reflem{lem:QDAinduction} and \reflem{lem:alphinduction}.

By induction on $\abs{\alpha(\tau_s(w))}$ we show that for every word $w$ there exists a language $L_w$ with 
$w \in L_w \subseteq h^{-1}h(w)$ with the desired property such that the number of products is bounded by a 
function depending on $h$ and $\alpha(\tau_s(w))$, but neither on $w$ nor on $\abs{w}$. 
In particular, there are only finitely many such languages $L_w$. 
Moreover, we ensure $\abs{v} \equiv \abs{w} \bmod s$ and $L_v = L_w$ for all $v \in L_w$. 
In addition, if $m > 2$, then $v \in L_w$ implies  $\alpha(\tau_s(v)) = \alpha(\tau_s(w))$. 
Note that $\abs{\alpha(\tau_s(w))} = \abs{\alpha(\tau_{j,s}(w))}$ for all integers $j$.

If $\alpha(\tau_s(w)) = \emptyset$, then $w = \varepsilon$ and we set $L_w = \smallset{\varepsilon}$. 
Let now $\alpha(\tau_s(w)) \neq \emptyset$ and consider the factorization 
\begin{math}
  w = w_1 a_1 \cdots w_k a_k w'
\end{math}
with 
\begin{equation*}
\alpha\big(\tau_{\abs{w_1 a_1 \cdots w_{i-1} a_{i-1}},s}(w_i)\big) \subsetneq 
\alpha\big(\tau_{\abs{w_1 a_1 \cdots w_{i-1} a_{i-1}},s}(w_i a_i)\big) = \alpha\big(\tau_s(w)\big)
\end{equation*}
such that $k \leq \abs{M}+1$ is 
minimal satisfying one (or both) of the following properties:
\begin{enumerate}
\item\label{aaa:h2p} $\alpha(\tau_{\abs{w_1 a_1 \cdots w_k a_k},s}(w')) \subsetneq \alpha(\tau_s(w))$.
\item\label{bbb:h2p} There exists $u \in (A^s)^*$ with 
$\alpha(\tau_s(w)) \subseteq \alpha(\tau_{\abs{w_1 a_1 \cdots w_k a_k},s}(u))$ such that $e=h(u)$ is idempotent and $h(w_1 a_1 \cdots w_k a_k) = h(w_1 a_1 \cdots w_k a_k u)$.
\end{enumerate}
If property~\ref{aaa:h2p}.\ holds, then we set
\begin{equation*}
  L_w = L_{w_1} a_1 \cdots L_{w_k} a_k L_{w'}
\end{equation*}
and this yields $w \in L_w \subseteq h^{-1}h(w)$. 
If property~\ref{aaa:h2p}.\ does not hold for all $k \leq \abs{M} + 1$, then we can consider the factorization $w = w_1 a_1 \cdots w_{\abs{M}+1} a_{\abs{M}+1} w''$. By the pigeonhole principle there exist $j<k\leq\abs{M}+1$ such that $h(w_1 a_1 \cdots w_j a_j) = h(w_1 a_1 \cdots w_k a_k)$. In this case we set $u = (w_{j+1} a_j \cdots w_k a_k)^{ns}$ for some integer $n$ such that $h(u)$ is idempotent. Therefore, we can assume that property~\ref{bbb:h2p}.\ holds and that property~\ref{aaa:h2p}.\ does not hold.
 Write $w' = x y$ such that $\abs{x} < s$ and $\abs{y} \equiv 0 \bmod s$.
If $m > 2$ then we set
\begin{equation*}
  L_w = L_{w_1} a_1 \cdots L_{w_k} a_k 
  \hspace*{1pt} x \hspace*{1pt} g^{-1}g(y).
\end{equation*}
By definition, we have $w \in L_w$. Let $v \in L_w$ and write $v = v_1 a_1 \cdots v_k a_k x v'$ with $v_i \in L_{w_i}$ and $h(v') \sim_K h(y)$. In particular, $h(v_i) = h(w_i)$.
We choose some word $z$ such that $\alpha(\tau_s(xyz)) \subseteq \alpha(\tau_s(u))$ and $\abs{xyz} \equiv 0 \bmod s$, \ie, we pad $x y$ to an $s$-divisible length.
Then we have $h(xyz) \in M^{\stab}_e$ and thus $eh(xyz)e = e$. We deduce $e \hspace*{1pt}h(xy) \Req e$ and hence, by $h(xy) \sim_K h(xv')$, 
we have $e \hspace*{1pt} h(xy) = e \hspace*{1pt} h(xv')$. It follows
\begin{align*}
h(w) &= h(w_1a_1w_2a_2 \cdots w_k a_k x y) 
\\ &= h(w_1a_1w_2a_2 \cdots w_k a_k )e h(x y) 
\\ &= h(v_1a_1v_2a_2 \cdots v_k a_k )e h(xv') = h(v)\,.
\end{align*}
This shows $L_w \subseteq h^{-1}h(w)$. 
The remaining case of the construction is $m=2$ in the situation where  property~\ref{bbb:h2p}.\ holds and property~\ref{aaa:h2p}.\ does not hold.
Let 
\begin{equation*}
  A_i = \set{a}{\exists p,q \colon w = paq \text{ and } \abs{p} \equiv i-1 \bmod s}
\end{equation*}
be the set of letters which appear in $w$ at a position congruent modulo $i$ and 
let $j \in \smallset{1,\ldots,s}$ satisfy $j \equiv \abs{w_1 a_1 \cdots w_k a_k x} \bmod s$. 
Then we set
\begin{equation*}
  L_w = L_{w_1} a_1 \cdots L_{w_k} a_k x (A_{j+1} \cdots A_s A_1 \cdots A_j)^*
\end{equation*}
Again, we trivially have $w \in L_w$. Let $v \in L_w$ and write $v = v_1 a_1 \cdots v_k a_k x v'$ with $v_i \in L_{w_i}$ and $v' \in (A_{j+1} \cdots A_s A_1 \cdots A_j)^*$. In particular, $h(v_i) = h(w_i)$.
As before, we choose some word $z$ such that $\alpha(\tau_s(xyz)) \subseteq \alpha(\tau_s(u))$ and $\abs{xyz} \equiv 0 \bmod s$, \ie, we pad $x y$ to an $s$-divisible length. Then we have $h(xyz) \in M^{\stab}_e$ and thus $eh(xyz)e = e$. We deduce $e h(xy) \Req e h(x)$. This implies $e h(xy) \Reqs eh(x)$ by Lemma~\ref{lem:RRst}. Since $M$ is $\Reqs$-trivial, we conclude $e h(xy) = e h(x)$. A similar reasoning shows $e h(x v') = eh(x)$. As in the previous case, this yields $h(w) = h(v)$.

Note that all products are $s$-modularly deterministic.
Moreover, in any case if $v \in L_w$, then $v$ admits an equivalent factorization as $w$; and this yields $L_v = L_w$. In the case of property~\ref{aaa:h2p}, this immediately follows by induction whereas the second case also relies on the fact that $g$ preserves the alphabetic information for words of length divisible by $s$. The prefix $L_{w_1} a_1$ of $L_w$ ensures that the alphabet $\alpha(\tau_s(v))$ is never too small for any word $v \in L_w$.
This shows that the union $\bigcup_{w \in L} L_w$ is disjoint and finite, and it coincides with $L$.
%
\end{proof}

\begin{lemma}\label{lem:product2FO2}
Let $L \subseteq A^*$ be expressible in the closure of languages $(A_1 \cdots A_n)^*$ for $A_i \subseteq A$
under finite union and modularly (co{-})deterministic products. Then $L$ is definable in $\FO^2[{<},\MOD]$.
\end{lemma}

\begin{proof}
The proof is by induction on the expression for $L \subseteq A^*$.
The language $(A_1 \cdots A_n)^*$ is defined by\,
$\LEN^n_0 \land \forall x : \bigwedge_{i=1}^n (x \equiv i \bmod n \limplies x \in A_i)$.
Thus consider a modularly deterministic product $LaK$ such that the letter $a \in A$ is at position $i \bmod n$ 
and there is no such $a$ at a position $j$ with $j \equiv i \bmod n$ in any word of $L$. 
We may assume, by using multiples of $n$, that $L$ and $K$ are expressible as $\FO^2[{<},\MOD^n]$-formulas. 
Let
\begin{align*}
\varrho(y) = \exists x\colon (y<x) \land (\forall y\colon &\lambda(x) = a \land x\equiv i \bmod n 
\\ &\land (\lambda(y) = a \land 
y\equiv i \bmod n \rightarrow x\leq y)).
\end{align*}
Then $\varrho(y)$ is used to check if $y$ is left of the position of $a$.
We can relativize $L$ and $K$ using the formula $\varrho$. 
Let $\varphi$ be an $\FO^2[{<},\MOD^n]$ formula with $L(\varphi) = L$. 
We inductively define $\varphi_{<a,i}$ which is true on the 
deterministic factorization $w = uav$ if $\varphi$ is true on $u$. 
Let 
\begin{align*}
&(\exists y: \tilde\varphi)_{<a,i} = \exists y (\varrho(y) \land \tilde\varphi_{<a,i})  \quad\quad\quad 
&&(\forall y: \tilde\varphi)_{<a,i} = \forall y (\varrho(y) \rightarrow \tilde\varphi_{<a,i})\\
&(\lambda(y) = a)_{<a,i} = \lambda(y) = a  &&(\MOD^n_j(y))_{<a,i} = \MOD^n_j(y)\\
& (\hat \varphi \land \tilde \varphi)_{<a,i} = \hat\varphi_{<a,i} \land \tilde\varphi_{<a,i} && 
(\lnot \tilde \varphi)_{<a,i} = \lnot \tilde \varphi_{<a,i}
\end{align*}
Symmetrically we can define $\psi_{>a,i}$ for a formula $\psi$ with $L(\psi) = K$, however we have to change the 
offset $(\MOD^n_j(y))_{>a,i} = \MOD^n_{j+i}(y)$. 
The product $LaK$ is now defined by the $\FO^2[{<},\MOD^n]$-formula 
$\exists x\colon (\lambda(x) = a \land x\equiv i \bmod s) \land \varphi_{<a,i} \land \psi_{>a,i}$. 
\end{proof}

We can now give a proof of the main result for $\FO^2[{<},\MOD]$.

\begin{proof}[Proof of Theorem~\ref{thm:FO2char}]
  Since $\FO^2[<]$ corresponds to $\varietyFont{DA}$, the equivalence of 
``\ref{aaa:FO2char}'' and ``\ref{bbb:FO2char}'' follows by \refprop{prp:fragment2starMOD}. 
``\ref{bbb:FO2char}''~implies~``\ref{ccc:FO2char}'': As $\FO^2[<]$ and $\Delta_2[<]$ have the same expressive power over finite words, 
 by \reflem{lem:starMODtoQS2} the syntactic homomorphism $h_L$ satisfies $eM^{\stab}_ee \leq e$ as well as $eM^{\stab}_ee \geq e$.
``\ref{ccc:FO2char}''~implies~``\ref{ddd:FO2char}'':
Let $N = (2^A)^s \times (\mathbb{Z}/s\mathbb{Z})$ be the monoid with the following multiplication $(A_1,\ldots,A_s,i) \cdot (B_1,\ldots,B_s,j) = (A_1 \cup B_{1+i \bmod s}, \ldots, A_s \cup B_{s+i \bmod s},\, i+j \bmod s)$. Let $\beta : A^* \to N$ be induced by $\beta(a) = (\smallset{a},\emptyset, \ldots,\emptyset, 1)$. Then $h : A^* \to \Synt(L) \times N, \; u \mapsto (h_L(u),\beta(u))$ satisfies the premise of \refprop{prp:hierarchy2product}. Therefore, $L$ is of the desired form.
``\ref{ddd:FO2char}''~implies~``\ref{aaa:FO2char}'' follows from \reflem{lem:product2FO2}.
\end{proof}

The following corollary is a consequence of Corollary~\ref{cor:delta2mod} and Theorem~\ref{thm:FO2char}.

\begin{corollary}\label{cor:fo2=delta2}
  Let $L \subseteq A^*$. Then $L$ is definable in $\FO^2[{<},\MOD]$ if and only if it is definable in $\Delta_2[{<},\MOD]$.
  \qed
\end{corollary}

We note that Corollary~\ref{cor:fo2=delta2} does not immediately follow from the fact that $\FO^2[{<}]$ and $\Delta_2[{<}]$ define the same languages over finite words since, in general, we have $(\varietyFont{V}*\VMOD) \cap (\varietyFont{W}*\VMOD) \neq (\varietyFont{V}\cap \varietyFont{W})*\VMOD$ for varieties $\varietyFont{V}$ and $\varietyFont{W}$.
The following example is due to Dartois and Paperman~\cite{DartoisPaperman13perscomm}.

\begin{example}\label{exa:intersectionMOD}
  Let $\varietyFont{R}$ be the full variety of $\greenR$-trivial monoids, and let $\varietyFont{L}$ be the full variety of $\greenL$-trivial monoids. The syntactic homomorphism $h_L$ of the language $L = (aa)^* (bb)^*$ is in both
$\varietyFont{R}*\VMOD$ and $\varietyFont{L}*\VMOD$ but not in $(\varietyFont{R}\cap \varietyFont{L})*\VMOD$. 
Furthermore, the stable monoid of $h_L$ is in $\varietyFont{R}\cap \varietyFont{L}$, \ie, the $\greenJ$-trivial monoids $\varietyFont{J} = \varietyFont{R}\cap \varietyFont{L}$ satisfy $\varietyFont{J} * \VMOD \neq \varietyFont{QJ}$.
\end{example}

\section*{Conclusion}

In Proposition~\ref{prp:fragment2starMOD} we have shown that the algebra operation $\varietyFont{V} \mapsto \varietyFont{V} * \VMOD$ corresponds to adding modular predicates to a given logical fragment. Unfortunately, this does not immediately help with decidability. In Theorem~\ref{thm:sigma2mod} we show that a language $L \subseteq A^*$ is $\Sigma_2[{<},\MOD]$-definable if and only if its syntactic homomorphism $h_L : A^* \to M$ satisfies $e M_e^{\stab} e \leq e$ for all idempotents $e$ in $M$. Since the latter property is decidable, one can effectively determine whether or not a given language is $\Sigma_2[{<},\MOD]$-definable. An important intermediate step in proving this decidability result is a characterization of the form $\varietyFont{V} * \VMOD$.

The characterization of the fragment $\Sigma_2[{<},\MOD]$ in Theorem~\ref{thm:sigma2mod} immediately leads to an algebraic counterpart of $\Delta_2[{<},\MOD]$. By definition, $\Delta_2[{<},\MOD]$ is the largest subclass of the $\Sigma_2[{<},\MOD]$-definable languages which is closed under complementation. We use this characterization of $\Delta_2[{<},\MOD]$ for showing that the two-variable fragment $\FO^2[{<},\MOD]$ has the same expressive power. Our proof yields two by-products. First, we characterize the $\FO^2[{<},\MOD]$-definable languages in terms of modularly deterministic and co-deterministic products which generalizes a result of Dartois and Paperman. The second by-product is another proof for the decidablity of $\FO^2[{<},\MOD]$; this was first shown by Dartois and Paperman~\cite{dartoispaperman13:short} using the characterization $\varietyFont{QDA}$, \ie, a language is $\FO^2[{<},\MOD]$-definable if and only if its stable monoid is in $\DA$. For $\Sigma_2[{<},\MOD]$ it is still open whether definability only depends on its stable monoid.

\section*{Acknowledgements}

We thank Luc Dartois, Alexander Lauser, and Charles Paperman for several fruitful discussions on modular predicates, we thank Volker Diekert for discussions on the proof of Lemma~\ref{lem:QDA2hierarchy}, and we are grateful to the anonymous referees for numerous suggestions which helped to improve the presentation of this paper.

{\small

\begin{thebibliography}{10}

\bibitem{arf91tcs}
M.~Arfi.
\newblock Op{\'e}rations polynomiales et hi{\'e}rarchies de concat{\'e}nation.
\newblock {\em Theoretical Computer Science}, 91(1):71--84, 1991.

\bibitem{BarringtonCST92jcss}
D.~A.~M. Barrington, K.~J. Compton, H.~Straubing, and D.~Th{\'e}rien.
\newblock Regular languages in {NC}$^{1}$.
\newblock {\em J. Comput. Syst. Sci.}, 44(3):478--499, 1992.

\bibitem{cps06:short}
L.~Chaubard, J.-{\'E}. Pin, and H.~Straubing.
\newblock Actions, wreath products of $\mathcal{C}$-varieties and concatenation
  product.
\newblock {\em Theor.\ Comput.\ Sci.}, 356:73--89, 2006.

\bibitem{cps06lics:short}
L.~Chaubard, J.-{\'E}. Pin, and H.~Straubing.
\newblock First order formulas with modular predicates.
\newblock In {\em LICS 2006, Proceedings}, pages 211--220. IEEE Computer
  Society, 2006.

\bibitem{DartoisPaperman13perscomm}
L.~Dartois and C.~Paperman.
\newblock Personal communication, 2013.

\bibitem{dartoispaperman13:short}
L.~Dartois and C.~Paperman.
\newblock Two-variable first order logic with modular predicates over words.
\newblock In {\em STACS 2013, Proceedings}, volume~20 of {\em LIPIcs}, pages
  329--340. Dagstuhl Publishing, 2013.

\bibitem{DiekertGastinKufleitner08:short}
V.~Diekert, P.~Gastin, and M.~Kufleitner.
\newblock A survey on small fragments of first-order logic over finite words.
\newblock {\em Int.\ J.\ Found.\ Comput.\ Sci.}, 19(3):513--548, 2008.

\bibitem{ei03}
Z.~{\'{E}}sik and M.~Ito.
\newblock Temporal logic with counting and the degree of aperiodicity of finite
  automata.
\newblock {\em Acta Cybernetica}, 16:1--28, 2003.

\bibitem{EsikLarsen03rairo}
Z.~{\'E}sik and K.~G. Larsen.
\newblock Regular languages definable by {L}indstr{\"o}m quantifiers.
\newblock {\em RAIRO, Inform. Th{\'e}or. Appl.}, 37(3):179--241, 2003.

\bibitem{kam68:short}
J.~A.~W. Kamp.
\newblock {\em Tense Logic and the Theory of Linear Order}.
\newblock PhD thesis, University of California, 1968.

\bibitem{krt68arbib8:short}
K.~Krohn, J.~L. Rhodes, and B.~Tilson.
\newblock Homomorphisms and semilocal theory.
\newblock In {\em Algebraic Theory of Machines, Languages, and Semigroups},
  chapter~8, pages 191--231. Academic Press, 1968.

\bibitem{KufleitnerLauser12mfcs:short}
M.~Kufleitner and A.~Lauser.
\newblock The join levels of the {T}rotter-{W}eil hierarchy are decidable.
\newblock In {\em MFCS 2012, Proceedings}, volume 7464 of {\em LNCS}, pages
  603--614. Springer, 2012.

\bibitem{KufleitnerL12icalp:short}
M.~Kufleitner and A.~Lauser.
\newblock Lattices of logical fragments over words.
\newblock In {\em ICALP 2012, Proceedings Part II}, volume 7392 of {\em LNCS},
  pages 275--286. Springer, 2012.

\bibitem{kw12}
M.~Kufleitner and P.~Weil.
\newblock On logical hierarchies within $\mathrm{FO}^{2}$-definable languages.
\newblock {\em Logical Methods in Computer Science}, 8(3), 2012.

\bibitem{mp71}
R.~McNaughton and S.~Papert.
\newblock {\em Counter-Free Automata}.
\newblock The MIT Press, Cambridge, Mass., 1971.

\bibitem{pin86}
J.-{\'E}. Pin.
\newblock {\em {Varieties of Formal Languages}}.
\newblock North Oxford Academic, London, 1986.

\bibitem{pin95:short}
J.-{\'E}. Pin.
\newblock A variety theorem without complementation.
\newblock In {\em Russian Mathematics (Iz.\ VUZ)}, volume~39, pages 80--90,
  1995.

\bibitem{pw97:short}
J.-{\'E}. Pin and P.~Weil.
\newblock Polynomial closure and unambiguous product.
\newblock {\em Theory Comput.\ Syst.}, 30(4):383--422, 1997.

\bibitem{pw02ca2:short}
J.-{\'E}. Pin and P.~Weil.
\newblock Semidirect products of ordered semigroups.
\newblock {\em Commun. Algebra}, 30(1):149--169, 2002.

\bibitem{sch65sf:short}
M.~P. Sch{\"u}tzenberger.
\newblock On finite monoids having only trivial subgroups.
\newblock {\em Inf.\ Control}, 8:190--194, 1965.

\bibitem{sim75:short}
I.~Simon.
\newblock Piecewise testable events.
\newblock In {\em Autom.\ Theor.\ Form.\ Lang., 2nd GI Conf.}, volume~33 of
  {\em LNCS}, pages 214--222. Springer, 1975.

\bibitem{Sto74}
L.~J. Stockmeyer.
\newblock The complexity of decision problems in automata theory and logic.
\newblock {PhD} thesis, {TR} 133, M.I.T., Cambridge, 1974.

\bibitem{str81tcs:short}
H.~Straubing.
\newblock A generalization of the {S}ch{\"u}tzenberger product of finite
  monoids.
\newblock {\em Theor.\ Comput.\ Sci.}, 13:137--150, 1981.

\bibitem{str85jpaa}
H.~Straubing.
\newblock Finite semigroup varieties of the form {$\mathbf{V}\ast \mathbf{D}$}.
\newblock {\em Journal of Pure and Applied Algebra}, 36(1):53--94, 1985.

\bibitem{str94}
H.~Straubing.
\newblock {\em Finite Automata, Formal Logic, and Circuit Complexity}.
\newblock Birkh{\"a}user, Boston, Basel and Berlin, 1994.

\bibitem{str02:short}
H.~Straubing.
\newblock On logical descriptions of regular languages.
\newblock In {\em LATIN 2002, Proceedings}, volume 2286 of {\em LNCS}, pages
  528--538. Springer, 2002.

\bibitem{st03tocs:short}
H.~Straubing and D.~Th{\'e}rien.
\newblock Regular languages defined by generalized first-order formulas with a
  bounded number of bound variables.
\newblock {\em Theory Comput.\ Syst.}, 36(1):29--69, 2003.

\bibitem{stt95IC:short}
H.~Straubing, D.~Th{\'e}rien, and W.~Thomas.
\newblock Regular languages defined with generalized quantifiers.
\newblock {\em Inf.\ Comput.}, 118(2):289--301, 1995.

\bibitem{tt02:short}
P.~Tesson and D.~Th{\'e}rien.
\newblock Diamonds are forever: {T}he variety $\mathrm{DA}$.
\newblock In {\em Semigroups, Algorithms, Automata and Languages 2001,
  Proceedings}, pages 475--500. World Scientific, 2002.

\bibitem{the81tcs:short}
D.~Th{\'e}rien.
\newblock Classification of finite monoids: {T}he language approach.
\newblock {\em Theor.\ Comput.\ Sci.}, 14(2):195--208, 1981.

\bibitem{tw98stoc:short}
D.~Th{\'e}rien and {\Th}.~Wilke.
\newblock Over words, two variables are as powerful as one quantifier
  alternation.
\newblock In {\em STOC 1998, Proceedings}, pages 234--240. ACM Press, 1998.

\bibitem{tho82:short}
W.~Thomas.
\newblock Classifying regular events in symbolic logic.
\newblock {\em J.\ Comput.\ Syst.\ Sci.}, 25:360--376, 1982.

\bibitem{wei11phd}
{\Ph}.~Weis.
\newblock {\em Expressiveness and succinctness of first-order logic on finite
  words}.
\newblock PhD thesis, University of Massachusetts Amherst, 2011.

\end{thebibliography}
\newcommand{\Ju}{Ju}\newcommand{\Ph}{Ph}\newcommand{\Th}{Th}\newcommand{\Ch}{Ch}\newcommand{\Yu}{Yu}\newcommand{\Zh}{Zh}\newcommand{\St}{St}\newcommand{\curlybraces}[1]{\{#1\}}

}

\newpage

\appendix

\section{Appendix: Semidirect products}
\label{app:semidirect}

In this appendix, we give an approach to the semidirect product $\varietyFont{V} * \VMOD$ where $\varietyFont{V}$ is an arbitrary positive variety of finite monoids and $\VMOD$ is some particular class of homomorphisms. Below, we show that the usual definition of $\varietyFont{V} * \VMOD$ and the definition used in this paper are equivalent. This can be seen as a variant of the so-called \emph{wreath product principle}.
An instance of $\varietyFont{V} * \VMOD$ with $\varietyFont{V}$ being a positive variety was already  studied by Chaubard, Pin, and Straubing~\cite{cps06lics:short}, but the proof of the wreath product principle given in~\cite{cps06lics:short} is only stated for full varieties~$\varietyFont{V}$. Pin and Weil studied semidirect products $\varietyFont{V} * \varietyFont{W}$ of varieties $\varietyFont{V}$ and $\varietyFont{W}$ such that $\varietyFont{V}$ is a positive variety~\cite{pw02ca2:short}. On the other hand, semidirect products $\varietyFont{V} * \varietyFont{W}$ with $\varietyFont{W}$ being a class of homomorphism were introduced by Chaubard, Pin, and Straubing~\cite{cps06:short}, see also~\cite{ei03}. The case $\varietyFont{V} * \varietyFont{W}$ where $\varietyFont{V}$ is a positive variety and where $\varietyFont{W}$ is a class of homomorphisms can therefore be seen as a conjunction of~\cite{cps06:short} and~\cite{pw02ca2:short}. We restrict ourselves to the case $\varietyFont{W} = \VMOD$.

We introduce semidirect products in terms of wreath products. Let $N$ and $K$ be monoids such that $N$ is ordered. Then the \emph{wreath product} $N \wr K$ is the set $N^K \times K$ with the composition
\begin{equation*}
  (f_1,k_1)(f_2,k_2) = (f,k_1 k_2) \; \text{ with } \; f(k) = f_1(k) f_2(k k_1).
\end{equation*}
The order on $N \wr K$ is defined by
\begin{equation*}
  (f_1,k_1) \leq (f_2,k_2) \; \text{ if } \; k_1 = k_2 \text{ and } f_1(k) \leq f_2(k) \text{ for all } k \in K.
\end{equation*}
Let $\varietyFont{V}$ be a class of finite ordered monoids and let $\varietyFont{W}$ be a class of homomorphisms of the form $h : A^* \to K$, so-called \emph{stamps}. A surjective homomorphism $h : A^* \to M$ belongs to the \emph{semidirect product} $\varietyFont{V} * \varietyFont{W}$ if there exists a homomorphism $\hat{h} : A^* \to N \wr K$ such that
\begin{itemize}
\item $N \in \varietyFont{V}$, 
\item $\pi_2 \circ \hat{h} : A^* \to K$ is a homomorphism in $\varietyFont{W}$, and
\item the following implication holds for all $u,v \in A^*$:
\begin{equation*}
  \hat{h}(u) \leq \hat{h}(v) \;\Rightarrow\; h(u) \leq h(v).
\end{equation*}
\end{itemize}
Here, $\pi_i$ denotes the projection to the $i$-th component. Let $\VMOD$ be the class of all homomorphism $h : A^* \to \mathbb{Z}/n\mathbb{Z}$ such that $h(a) = h(b)$ for all letters $a,b \in A$. As usual, $\mathbb{Z}/n\mathbb{Z}$ denotes the cyclic group of order $n$, implemented using addition modulo $n$.

\begin{proposition}\label{prp:semidirect}
  Let $\varietyFont{V}$ be a positive variety of finite monoids and let $h : A^* \to M$ be a homomorphism onto a finite ordered monoid $M$. We have $h \in \varietyFont{V} * \VMOD$ if and only if there exists an integer $n > 0$ and a homomorphism $g : T_n(A)^* \to N$ with $N \in \varietyFont{V}$ satisfying
  \begin{equation*}
    g\big(\tau_n(u)\big) \leq g\big(\tau_n(v)\big) 
  \;\Rightarrow\; h(u) \leq h(v)
  \end{equation*}
  for all $u,v \in A^*$ with $\abs{u} \equiv \abs{v} \bmod n$.
\end{proposition}

\begin{proof}
  For the implication from left to right let $\hat{h} : A^* \to N \wr (\mathbb{Z}/n\mathbb{Z})$ be a homomorphism with $N \in \varietyFont{V}$ and $\pi_2 \circ \hat{h}(a) = d$ for all $a \in A$, and suppose that $\hat{h}(u) \leq \hat{h}(v)$ implies $h(u) \leq h(v)$ for all $u,v \in A^*$. Let $\hat{h}(a) = (f_a,d)$ for $f_a \in N^{\mathbb{Z}/n\mathbb{Z}}$. For a function $f \in N^{\mathbb{Z}/n\mathbb{Z}}$ and $i \in \mathbb{Z}/n\mathbb{Z}$ we define $i \cdot f \in N^{\mathbb{Z}/n\mathbb{Z}}$ by 
  \begin{equation*}
    (i \cdot f)(k) = f(k+i).
  \end{equation*} 
  Using this notation we define $g : T_n(A)^* \to N^{\mathbb{Z}/n\mathbb{Z}}$ by $g(a,i) = (i-1)d \cdot f_a$ for $(a,i) \in T_n(A)$. The composition in $N^{\mathbb{Z}/n\mathbb{Z}}$ is the componentwise composition of $N$; we have $N^{\mathbb{Z}/n\mathbb{Z}} \in \varietyFont{V}$ since $\varietyFont{V}$ is closed under direct products. For every word $u = a_1 \cdots a_k$ with $a_i \in A$, the definition of the wreath product and the definition of $g$ yields
  \begin{equation*}
    \pi_1 \circ \hat{h}(u)
    = f_{a_1} \, \big( d \cdot f_{a_2} \big) \, \big( 2d \cdot f_{a_3} \big) \, \cdots \, \big( (k-1)d \cdot f_{a_k} \big)
    = g\big(\tau_n(u)\big).
  \end{equation*}
  Consider words $u,v \in A^*$ with $\abs{u} \equiv \abs{v} \bmod n$ and $g\big(\tau_n(u)\big) \leq g\big(\tau_n(v)\big)$. Then 
  \begin{equation*}
    \hat{h}(u) = \big( g(\tau_n(u)),\, d\abs{u} \bmod n \big)
     \leq \big( g(\tau_n(v)),\, d\abs{v} \bmod n \big)
     = \hat{h}(v)
  \end{equation*}
  and thus $h(u) \leq h(v)$.
  
  For the implication from right to left let $n$, $N$, and $g$ be as in the statement of the proposition. For every letter $a \in A$ we define $f_a \in N^{\mathbb{Z}/n\mathbb{Z}}$ by $f_a(k) = g(a,k+1)$. This yields the homomorphism $\hat{h} : A^* \to N \wr (\mathbb{Z}/n\mathbb{Z})$ with
  \begin{equation*}
    \hat{h}(a) = (f_a,1).
  \end{equation*}
  As before, for a function $f \in N^{\mathbb{Z}/n\mathbb{Z}}$ and $i \in \mathbb{Z}/n\mathbb{Z}$ we define $i \cdot f \in N^{\mathbb{Z}/n\mathbb{Z}}$ by $(i \cdot f)(k) = f(k+i)$.
  Consider a word $u = a_1 \cdots a_k$ with $a_i \in A$. Then $\hat{h}(u) = (f,\abs{u} \bmod n)$ with
  \begin{equation*}
    f = f_{a_1} (1 \cdot f_{a_2}) (2 \cdot f_{a_3}) \cdots \big((k-1) \cdot f_{a_k} \big).
  \end{equation*}
  By definition of the functions $f_{a_i}$ we have $f(0) = g(\tau_n(u))$. Therefore, for all words $u,v \in A^*$, if $\hat{h}(u) \leq \hat{h}(v)$, then $\abs{u} \equiv \abs{v} \bmod n$ and $g(\tau_n(u)) \leq g(\tau_n(v))$, and thus $h(u) \leq h(v)$.
\end{proof}

\end{document}